\documentclass[11pt,a4paper]{article}
\pdfoutput=1
\usepackage[T1]{fontenc}
\usepackage[utf8]{inputenc}
\usepackage{lmodern,microtype}
\usepackage[margin=25mm]{geometry}
\usepackage{amsmath,amssymb,amsthm,mathtools}
\usepackage{graphicx,booktabs}
\usepackage[colorlinks=true,linkcolor=blue,citecolor=blue,urlcolor=blue]{hyperref}
\usepackage{enumitem}
\hypersetup{pdftitle={Analytic Brill wave initial data with trapped surfaces},pdfauthor={Piotr Koc}}
\numberwithin{equation}{section}
\newtheorem{theorem}{Theorem}[section]
\newtheorem{proposition}[theorem]{Proposition}
\newtheorem{lemma}[theorem]{Lemma}
\theoremstyle{remark}\newtheorem{remark}[theorem]{Remark}
\newcommand{\R}{\mathbb R}
\newcommand{\D}{\mathcal D^{1,2}(\R^3)}
\newcommand{\dd}{\mathrm d}
\newcommand{\sech}{\operatorname{sech}}
\newcommand{\atanh}{\operatorname{artanh}}

\newcommand{\epsn}{\epsilon_n}
\newcommand{\half}{\tfrac12}
\title{Analytic Brill wave initial data with trapped surfaces}
\author{Piotr Koc\\[2pt]\small Independent researcher, Krak\'ow, Poland\\
\small\texttt{piotr@koc.pl}}
\date{21 November 2026}
\begin{document}
\maketitle
\begin{abstract}
Since Brill's work of 1959, the conformal factor of finite-amplitude Brill
waves, the axisymmetric and time-symmetric vacuum initial data on $\R^3$
that describe pure gravitational radiation, has been obtained only
numerically, up to the most recent studies. Here an explicit one-parameter
family of such data is constructed; the data are real analytic and
asymptotically flat, with trivial topology and a single end. The free
metric function is a closed-form combination of
elementary functions and a dilogarithm. The data have the toroidal
conformal symmetry of Beig and Husa, and a change of variable of Natanzon
type reduces every mode equation to the Gauss hypergeometric equation, so
that the conformal factor is a series with explicit Gamma-function
coefficients. A ground-state identity gives a sharp existence
criterion: complete, asymptotically flat data exist precisely for
$k<k_c=0.80480168\ldots$, and the explicit series is identified with the
unique variational solution. The ADM mass is an absolutely convergent
digamma series, with $m/a=\pi k^2/10+O(k^3)$ for small amplitudes and an
explicit residue at its critical divergence. Sufficiently strong waves
contain trapped surfaces: it is proved that for $k$ close to $k_c$ the
rescaled conformal factor converges to an explicit function and explicit
tori and spheres become outer trapped, so that the data contain an
apparent horizon. A numerical search, restricted to equatorially symmetric
star-shaped surfaces, finds marginally trapped (minimal) spheres, which
appear in a fold near $k=0.7670$. A scaling symmetry of the mode
problem generates further solvable families, whose simplest member has
the elementary seed function $q=\frac12\log[1+4a^2\rho^2/(\rho^2+z^2+a^2)^2]$.
No thin-torus or weak-field approximation is used.
\end{abstract}

\section{Introduction}
In 1959 Dieter Brill, in his Princeton dissertation and the paper based
on it \cite{Brill}, studied axially symmetric, time-symmetric vacuum
initial data for the Einstein equations and proved that their mass is
positive. These data are the simplest setting in which pure
gravitational radiation can be studied exactly, beyond the linear approximation,
and the standard setting for the old question whether gravitational
waves alone can form a black hole \cite{Wheeler}; yet, to the author's
knowledge, no explicit solution has been known. In cylindrical
coordinates the physical three-metric is
\begin{equation}\label{metric}
 g=\Phi^4\bigl[e^{2q}(\dd\rho^2+\dd z^2)+\rho^2\dd\varphi^2\bigr],
 \qquad K_{ij}=0.
\end{equation}
Since $K_{ij}=0$, the momentum constraint holds identically, and the
Hamiltonian constraint ${}^{3}\!R=0$, the vanishing of the scalar
curvature of $g$, reduces for axisymmetric $\Phi$ to a Schr\"odinger
equation at zero energy for the conformal factor, with potential $-V$,
and a Poisson equation in the $(\rho,z)$ plane that links $V$ to the
seed function:
\begin{equation}\label{constraints}
 \Delta_3\Phi+V\Phi=0,\qquad V=\frac14\Delta_2q,
 \qquad\Delta_2=\partial_\rho^2+\partial_z^2.
\end{equation}
The function $q$ is freely prescribed, subject to the boundary conditions
that follow from regularity of the metric on the axis and from asymptotic
flatness \cite{Brill}:
\begin{equation}\label{bcq}
 q=0,\quad\partial_\rho q=0\quad\text{on the axis }\rho=0;\qquad
 q=O(r^{-2})\quad(r\to\infty),
\end{equation}
where $r^2=\rho^2+z^2$. The conformal factor must be positive and regular
on the axis, and
\begin{equation}\label{bcPhi}
 \partial_\rho\Phi=0\quad\text{at }\rho=0;\qquad
 \Phi=1+\frac m{2r}+O(r^{-2})\quad(r\to\infty),
\end{equation}
where $m$ is the ADM mass. The data considered below are, in addition,
symmetric under $z\mapsto-z$. Although the equation for $\Phi$ is linear
for fixed $q$, the resulting physical metric is not a linearized gravitational
field. For finite amplitudes the conformal factor of Brill data on $\R^3$
has so far been obtained only numerically, starting with Eppley
\cite{Eppley} and continuing with a variety of methods
\cite{Gentle,KKS,Korobkin,deOliveira}, in particular in studies of the
collapse of gravitational waves
\cite{GarfinkleDuncan,Sorkin,Hilditch13,Hilditch16,Hilditch17}.
This is still the state of the art. In recent work the conformal factor
is computed numerically for Brill waves with slow fall-off \cite{BGW25}
and for data with prescribed asymptotics \cite{BGIMW26}; the dependence
of the ADM mass on the wave amplitude is computed numerically and
fitted \cite{AranhaOliveira}; and the range of amplitudes in which the
constraint has a solution, together with the occurrence of apparent
horizons, is mapped numerically \cite{MastersonHobill}.
General properties of such data were derived by indirect arguments
\cite{Wheeler,Holz,BeigOM}, and existence itself follows from general
theory \cite{CantorBrill}.

A classical question about such data is whether pure gravitational
radiation, when sufficiently concentrated, produces trapped surfaces.
Eppley \cite{Eppley} found apparent horizons in strong Brill waves
numerically; Beig and \'O Murchadha \cite{BeigOM} proved that trapped
surfaces must form as a critical amplitude is approached; the Penrose
inequality was tested on such configurations \cite{KKS}; the parameter
regions with and without apparent horizons were mapped numerically
\cite{MastersonHobill}; and trapped
surfaces of toroidal topology were studied in \cite{Husa96,KarkowskiTori}.
The explicit family constructed here allows this question to be examined
on exact data. Section~\ref{sec_proof} proves that the data contain outer
trapped surfaces, and hence an apparent horizon, for $k$ sufficiently
close to $k_c$; Section~\ref{sec_trapped} locates marginally trapped
spheres numerically, together with the outer trapped surfaces they
enclose.

Beig and Husa \cite{BeigHusa} introduced a class with a
$U(1)\times U(1)$ group of conformal isometries. Fourier decomposition reduces
the constraint to independent ordinary differential equations. Their
axisymmetric subclass contains the data studied here. Husa subsequently developed
numerical constructions \cite{Husa98}. His toroidal inner-boundary problem also
has an explicit limiting solution, expressed through Bessel functions, but that
limit develops a curvature singularity \cite{Husa96}. Analytic linearized Brill
data have been constructed separately \cite{LinearBrill}. These results must be
distinguished from an everywhere regular, finite-amplitude family on all of
$\R^3$.

The present contribution is an explicitly solvable analytic family. Its
starting point is a reversal of the usual order of the problem. In the
approach of all the numerical work cited above one prescribes $q$,
computes the potential $V=\Delta_2q/4$, and solves the linear equation
\eqref{constraints} for $\Phi$. Here the potential is chosen first, in
such a way that the equation for $\Phi$ becomes solvable, and $q$ is
recovered from $V$ afterwards by integration. In this order the boundary
conditions \eqref{bcq} are not automatic: they become a condition on the
potential, and this condition fixes the one constant that solvability
leaves free. This strategy was introduced in \cite{Koc21}, with a
piecewise constant potential confined to a thin torus and a matching
carried out to leading order in the thickness of the torus; the results
of that approximate construction are superseded by the exact family
presented here. The separation of
variables is the one just described. The radial equations belong to the
hypergeometric class described by Natanzon \cite{Natanzon}, and the
connection formula then determines every mode normalization and the
mass.

Throughout, ``explicit'' has a definite meaning: $q$ is a closed-form
combination of elementary functions and a dilogarithm, elementary for
$k=-3$, and $\Phi$ is a series of
hypergeometric functions whose coefficients are closed-form Gamma-function
expressions. To the author's knowledge, these are the first explicit
vacuum initial data for the Einstein equations that are smooth on all of
$\R^3$, asymptotically flat and of positive mass; all such data obtained
so far, from Brill's original work to the most recent studies
\cite{Brill,Eppley,Gentle,KKS,Korobkin,deOliveira,GarfinkleDuncan,Sorkin,Hilditch13,Hilditch16,Hilditch17,BGW25,BGIMW26,AranhaOliveira,MastersonHobill},
have been constructed numerically. The classical explicit vacuum data, namely the Schwarzschild
metric in isotropic coordinates and the Brill--Lindquist and Misner
families \cite{BrillLindquist,Misner}, are likewise analytic and free of
singularities, but they live on manifolds with several asymptotic ends:
the black holes are built in through the topology. The present data are
pure gravitational waves on a manifold with trivial topology and a single
end.

Every quantity addressed by the numerical studies cited above is
obtained here in exact form, for the seed function \eqref{qclosed}. The
conformal factor is the explicit series \eqref{Phiseries}. The ADM mass
is the convergent series \eqref{mass} on the whole interval of existence,
with the exact weak-field coefficient \eqref{weakmass} and the exact pole
\eqref{masspole}; no fit is involved. The limit of existence is not
bracketed numerically: it is the root $k_c$ of the elementary equation
\eqref{kceq}, and both existence for every $k<k_c$ and non-existence
beyond it are proved (Theorem~\ref{existence}). The formation of trapped
surfaces near this limit is proved as well, with the limiting geometry in
closed form (Theorem~\ref{trapped} and Remark~\ref{limitgeometry}). The
family therefore provides an exact benchmark for constraint solvers,
mass computations and horizon finders, including an elementary member,
\eqref{simplest}, whose seed function is rational-logarithmic in the
cylindrical coordinates.

The result is contained in a few formulas. The metric \eqref{metric} is
determined by two functions. The free function $q$ is given by
\eqref{qclosed}, with the toroidal coordinate $\nu(\sigma)$ from
\eqref{nusigma} and the constant $S(k)$ from \eqref{Sdef}; its elementary
form at $k=-3$ is \eqref{qminus3}. The conformal factor $\Phi$ is the
series \eqref{Phiseries} with the modes \eqref{explicit_modes} and the
parameters \eqref{ab}. The ADM mass is \eqref{mass}, and the family exists
for $k<k_c$ with $k_c$ defined by \eqref{kceq}. Finally, a scaling
symmetry of the mode problem (Section~\ref{sec_scaling}) shows that this
family is the first of a countable sequence of solvable families, and
produces the simplest explicit Brill wave: the elementary seed function
\eqref{simplest}, $q=\frac12\log(1+\tanh^2\nu)$, with a rational
potential and hypergeometric modes. Section~\ref{sec_curvature}
evaluates the Kretschmann scalar of the data, with closed forms on the
axis, on the circle and in the critical limit; in units of the mass the
curvature at the centre reaches $\mathcal Km^4=786$, a thousand times
its value on a Schwarzschild horizon, below the threshold of trapped
spheres.

The argument is organized so that existence precedes reconstruction. A
positive auxiliary function proves coercivity and identifies the exact
critical parameter. Fourier analysis of the resulting global solution then
justifies the explicit series, including its normalization and its mass.
This avoids inferring global regularity from pointwise numerical checks of
individual modes. The existence of trapped surfaces near the critical
parameter is proved from the same explicit modes
(Theorem~\ref{trapped}); the numerical trapped-surface calculations are
stated with their restricted geometric scope. Numerical checks of the formulas and a
description of the accompanying code are collected in
Appendix~\ref{app_verification}.

\section{Toroidal separation and endpoint conditions}
This section reduces the constraint \eqref{constraints} to ordinary
differential equations. Fix a length scale $a>0$. In the meridional half-plane $\rho\ge0$
introduce coordinates $(\nu,u)$ by
\begin{equation}\label{coords}
 \rho=\frac{a\sinh\nu}{d},\qquad
 z=\frac{a\sin u}{d},\qquad
 d(\nu,u)=\cosh\nu-\cos u,
\end{equation}
with $0\le\nu<\infty$ and $-\pi<u\le\pi$. The function $d(\nu,u)$ is
written $d$ for short throughout; it is not to be confused with the
differential, which is set upright. Together with the azimuth $\varphi$
these are toroidal coordinates on $\R^3$. Equivalently, $\rho+iz=a\coth[(\nu-iu)/2]$. The coordinate loci
are as follows. The axis $\rho=0$ is $\nu=0$. The circle $\rho=a$, $z=0$
is approached as $\nu\to\infty$. Spatial infinity is the single point
$(\nu,u)=(0,0)$, where $d=0$, and the origin is $(\nu,u)=(0,\pi)$. The
equatorial plane $z=0$ is $u=0$ outside the circle and $u=\pi$ inside
it, and $z$ has the sign of $u$. The surfaces $\nu=\mathrm{const}$ are
the tori $(\rho-a\coth\nu)^2+z^2=a^2/\sinh^2\nu$ around the circle, and
the surfaces $u=\mathrm{const}$ are the spherical caps
$\rho^2+(z-a\cot u)^2=a^2/\sin^2u$ bounded by it. With $r^2=\rho^2+z^2$,
three consequences of \eqref{coords} are used repeatedly:
\begin{equation}\label{toroidalrelations}
 \tanh\nu=\frac{2a\rho}{r^2+a^2},\qquad
 \frac{r^2}{a^2}=\frac{\cosh\nu+\cos u}{\cosh\nu-\cos u},\qquad
 d=\frac{2a^2}{\sqrt{[(\rho-a)^2+z^2][(\rho+a)^2+z^2]}} .
\end{equation}
The map $(\nu,u)\mapsto(\rho,z)$ is conformal, and the flat metric is
$(a^2/d^2)(\dd\nu^2+\dd u^2+\sinh^2\nu\,\dd\varphi^2)$.

We assume that $q$ depends on $\nu$ only, $q=q(\nu)$. By conformal
covariance, $\Delta_2=(d^2/a^2)(\partial_\nu^2+\partial_u^2)$, so with
the reduced potential $w\equiv(a^2/d^2)V$ the relation $V=\frac14\Delta_2q$
of \eqref{constraints} becomes
\begin{equation}\label{V}
 w=\frac14\,q''(\nu),\qquad{}'=\frac{\dd}{\dd\nu},
\end{equation}
so that $V=(d^2/a^2)\,w(\nu)$: a function of $\nu$ alone times $d^2$.
Conversely, $q$ is recovered from $w$ by two integrations,
\begin{equation}\label{qfromw}
 q(\nu)=4\int_0^\nu\dd\nu'\int_0^{\nu'}w(\nu'')\,\dd\nu'' ,
\end{equation}
where the two constants of integration have been fixed by
$q(0)=q'(0)=0$, the form that the boundary conditions \eqref{bcq} take
on the axis (see \eqref{bcqtor} below). This is the order in which the
functions are determined in Section~\ref{sec_family}: the reduced
potential $w$ is prescribed so that the mode equations derived below are
solvable, and $q$ follows from \eqref{qfromw}. The metric \eqref{metric}
takes the diagonal form
\begin{equation}\label{gtor}
 g=\frac{a^2\Phi^4}{d^2}
 \bigl[e^{2q}(\dd\nu^2+\dd u^2)+\sinh^2\nu\,\dd\varphi^2\bigr].
\end{equation}
For functions independent of $\varphi$, the flat Laplacian in toroidal
coordinates is
\begin{equation}\label{Laptor}
\begin{split}
 \Delta_3\Phi&=\frac{d^3}{a^2\sinh\nu}\left[
   \partial_\nu\Bigl(\frac{\sinh\nu}{d}\,\Phi_\nu\Bigr)
  +\partial_u\Bigl(\frac{\sinh\nu}{d}\,\Phi_u\Bigr)\right]\\
 &=\frac{d^2}{a^2}\Bigl[\Phi_{\nu\nu}+\coth\nu\,\Phi_\nu+\Phi_{uu}
  -\frac{\sinh\nu\,\Phi_\nu+\sin u\,\Phi_u}{d}\Bigr],
\end{split}
\end{equation}
the second form by $\partial_\nu(1/d)=-\sinh\nu/d^2$ and
$\partial_u(1/d)=-\sin u/d^2$. The factor $d^2$ in front of it is the
same factor that \eqref{V} puts in front of $w$. Dividing the
constraint \eqref{constraints} by $d^2/a^2$ therefore gives
\begin{equation}\label{constrainttoroidal}
 \Phi_{\nu\nu}+\coth\nu\,\Phi_\nu+\Phi_{uu}
 -\frac{\sinh\nu\,\Phi_\nu+\sin u\,\Phi_u}{d}+w(\nu)\,\Phi=0 .
\end{equation}
The coefficients of \eqref{constrainttoroidal} depend on $u$ only
through the single term with $1/d$; this is what makes the separation
possible. That term is removed by the substitution
\begin{equation}\label{PhiF}
 \Phi=\sqrt d\,F(\nu,u).
\end{equation}
A direct computation, which uses $\partial_\nu d=\sinh\nu$,
$\partial_ud=\sin u$ and
$\sinh^2\nu+\sin^2u=d\,(\cosh\nu+\cos u)$, gives the identity
\begin{equation}\label{sqrtd}
 \Delta_3\bigl(\sqrt d\,F\bigr)=\frac{d^{5/2}}{a^2}
 \left[F_{\nu\nu}+\coth\nu\,F_\nu+F_{uu}+\tfrac14F\right].
\end{equation}
The derivatives of $\sqrt d$ cancel the term of
\eqref{constrainttoroidal} that contains $1/d$, and the only trace of
$d$ inside the bracket is the constant $\tfrac14$. Since $V\Phi=(d^{5/2}/a^2)\,wF$ by \eqref{V} and
\eqref{PhiF}, the constraint \eqref{constraints} becomes
\begin{equation}
 \frac{d^{5/2}}{a^2}\left[F_{\nu\nu}+\coth\nu\,F_\nu+F_{uu}
       +\bigl(\tfrac14+w\bigr)F\right]=0 .
\end{equation}
As $d>0$ away from the point at infinity, the constraint
\eqref{constraints} is equivalent to
\begin{equation}\label{F_eq}
 F_{\nu\nu}+\coth\nu\,F_\nu+F_{uu}
       +\bigl(\tfrac14+w\bigr)F=0 ,
\end{equation}
an equation whose coefficients depend on $\nu$ only. For $w=0$ it is
solved by $F=1/\sqrt d$, which is the flat solution $\Phi=1$. For
$w=0$ the substitution \eqref{PhiF} is the classical $R$-separation of
Laplace's equation in toroidal coordinates \cite{MorseFeshbach}, whose
separated solutions are the toroidal harmonics
$\sqrt d\,P_{n-1/2}(\cosh\nu)\cos(nu)$ and
$\sqrt d\,Q_{n-1/2}(\cosh\nu)\cos(nu)$ \cite{AS,DLMF}. The identity
\eqref{sqrtd} shows that the same substitution survives the addition of
a potential, provided the potential has the form \eqref{V}.

For fixed $\nu>0$ one has $d>0$, so a smooth $\Phi$ gives a smooth
$2\pi$-periodic function $F=\Phi/\sqrt d$ of $u$. The potential
\eqref{V} is even in $u$, that is, symmetric under $z\mapsto-z$, and the
solution constructed below shares this symmetry
(Theorem~\ref{existence}); hence only cosines occur:
\begin{equation}\label{Fourier}
 F(\nu,u)=\sum_{n=0}^\infty f_n(\nu)\cos(nu).
\end{equation}
Because the coefficients of \eqref{F_eq} do not depend on $u$ and
$\partial_u^2\cos(nu)=-n^2\cos(nu)$, each coefficient satisfies an
ordinary differential equation,
\begin{equation}\label{modeeq}
 f_n''+\coth\nu\,f_n'+\bigl(\tfrac14-n^2+w\bigr)f_n=0.
\end{equation}
For $w=0$ this is Legendre's equation of degree $n-1/2$ in the variable
$\cosh\nu$. The boundary conditions, regularity at the circle and the
behaviour at the axis that corresponds to $\Phi\to1$ at infinity, are
imposed in Section~\ref{sec_modes}.

This is the Beig--Husa reduction in toroidal coordinates. In geometric
terms, put $\varrho=\tanh\nu$. Multiplying the seed metric, the bracket
in \eqref{gtor} times $a^2/d^2$, by $d^2/(a^2\cosh^2\nu)$ gives
\begin{equation}\label{compactmetric}
 \bar g=e^{2q}\left[\frac{\dd\varrho^2}{1-\varrho^2}
                  +(1-\varrho^2)\dd u^2\right]
                 +\varrho^2\dd\varphi^2
\end{equation}
on $S^3$. Its coefficients do not depend on $u$ and $\varphi$, so
$\partial_u$ and $\partial_\varphi$ are Killing fields of $\bar g$ and
conformal Killing fields of the seed metric; these are the two commuting
circle actions of \cite{BeigHusa}.

It remains to express the boundary conditions \eqref{bcq} for a seed
function $q(\nu)$. On the axis one has $\nu=0$ and
$\partial_\rho\nu\ne0$, so \eqref{bcq} requires
\begin{equation}\label{bcqtor}
 q(0)=0,\qquad q'(0)=0 .
\end{equation}
The decay condition in \eqref{bcq} then holds automatically, because
spatial infinity is the point $(\nu,u)=(0,0)$ of the axis. By
\eqref{toroidalrelations}, $\nu=2a\rho/r^2+O(r^{-3})$ there, and a
smooth even function $q$ with \eqref{bcqtor} has the expansion
\begin{equation}\label{q_asym}
 q=2w(0)\nu^2+O(\nu^4)
   =\frac{8a^2w(0)\rho^2}{r^4}+O(r^{-4}).
\end{equation}
In physical Cartesian coordinates this implies, with corresponding derivative
estimates,
\begin{equation}\label{V_asym}
 V=\frac{4a^2w(0)}{r^4}+O(r^{-6}).
\end{equation}
The leading term is isotropic, although that of $q$ is not.

The toroidal form adds one condition that has no counterpart in
\eqref{bcq}. The coordinates are singular on the circle $\rho=a$, $z=0$,
where $q$ must nevertheless be smooth. Near the circle
$\nu\simeq\log(2a/\delta)$, with $\delta$ the distance from it, so a
non-zero limit $q'(\infty)$ would give $q$ a logarithmic singularity.
Hence $q'(\infty)=0$. Since $q'(0)=0$ and $w=q''/4$, this imposes
\begin{equation}\label{balance}
 \int_0^\infty w(\nu)\,\dd\nu=0.
\end{equation}
The Euclidean volume element is
$a^3\sinh\nu\,d^{-3}\dd\nu\,\dd u\,\dd\varphi$, and
$\int_{-\pi}^{\pi}\dd u/d=2\pi/\sinh\nu$. With \eqref{V}, condition
\eqref{balance} is therefore the same as
\begin{equation}\label{zeromean}
 \int_{\R^3}V\,\dd^3x=4\pi^2a\int_0^\infty w(\nu)\,\dd\nu=0 ,
\end{equation}
the vanishing of the volume integral of the potential, which holds for
every regular Brill wave \cite{Wheeler}. The same condition fixed the
potential in the thin-torus construction \cite{Koc21}. It excludes, for
example, a nonzero profile proportional to $\sech^2\nu$ alone.

Regularity by itself does not require analyticity. In the class of
analytic functions considered here, the conditions above mean that $q$
is an even analytic function of $\nu$ with $q(0)=0$, and that at the
circle it is analytic in
\begin{equation}\label{x_circle}
 x=e^{-2\nu}=\frac{(\rho-a)^2+z^2}{(\rho+a)^2+z^2}.
\end{equation}

\section{An analytic hypergeometric family}\label{sec_family}
In accordance with the strategy described in the Introduction, the
reduced potential $w$ is now chosen so that every mode equation
\eqref{modeeq} is solvable in closed form. The seed function $q$ is
obtained afterwards from \eqref{qfromw}, and the boundary condition
\eqref{balance} at the circle fixes the one constant that solvability
leaves free.

Writing $f_n=g_n/\sqrt{\sinh\nu}$ transforms \eqref{modeeq} to
\begin{equation}\label{schrodinger}
 g_n''+\left[-n^2+\frac1{4\sinh^2\nu}+w\right]g_n=0.
\end{equation}
The new independent variable is a function $\sigma(\nu)$, $0\le\sigma<1$,
defined for a real parameter $k<1$ by
\begin{equation}\label{Rdef}
 \frac{\dd\sigma}{\dd\nu}=\frac{1-\sigma^2}{R(\sigma)},\qquad
 R(\sigma)=\sqrt{1-k+k\sigma^2},\qquad \sigma(0)=0 ,
\end{equation}
or explicitly, after integration,
\begin{equation}\label{nusigma}
 \nu(\sigma)=\atanh\frac{\sigma}{R}
             -\sqrt{k}\,\atanh\frac{\sqrt{k}\,\sigma}{R}.
\end{equation}
For $k<0$, the second term means
$+\sqrt{|k|}\arctan(\sqrt{|k|}\,\sigma/R)$. All formulas at $k=0$ are
understood by continuity; in particular, $\sigma=\tanh\nu$. The variable
runs from the axis, where $\sigma\simeq\nu/\sqrt{1-k}$, to the circle,
where $1-\sigma$ is proportional to $e^{-2\nu}$.

This choice of variable comes from the following consideration. In
\eqref{schrodinger} change the independent variable to some
$\zeta(\nu)$. The first-derivative term that this produces is removed by
rescaling the unknown, $g=(\zeta')^{-1/2}\,Y(\zeta)$, a Liouville
transformation. Then $Y$ again satisfies an equation in normal form,
\begin{equation}\label{normalform}
 \frac{\dd^2Y}{\dd\zeta^2}+I(\zeta)\,Y=0,\qquad
 I(\zeta)=\frac{Q-\frac12\{\zeta,\nu\}}{\zeta'^2},
 \qquad
 \{\zeta,\nu\}=\frac{\zeta'''}{\zeta'}-
                       \frac32\left(\frac{\zeta''}{\zeta'}\right)^2,
\end{equation}
where $Q$ is the bracket in \eqref{schrodinger} and $\{\zeta,\nu\}$ is
the Schwarzian derivative. Natanzon \cite{Natanzon} described the
changes of variable for which $I(\zeta)$ becomes the normal-form
invariant of the hypergeometric equation: they satisfy
$\zeta'^2=4\zeta^2(1-\zeta)^2/\mathcal R(\zeta)$ with $\mathcal R$ a
quadratic polynomial, and $Y$ is then a hypergeometric function
multiplied by powers of $\zeta$ and $1-\zeta$. The endpoint behaviour
required here, $\zeta\sim\nu^2$ at the axis and $1-\zeta\sim e^{-2\nu}$
at the circle, restricts $\mathcal R$ to the form $\zeta(1-k+k\zeta)$,
that is,
\begin{equation}\label{natanzon}
 \zeta'^2=\frac{4\zeta(1-\zeta)^2}{1-k+k\zeta},\qquad k<1 ,
\end{equation}
and \eqref{Rdef} is this equation written for $\sigma=\sqrt\zeta$. This
selects a useful family; the alternative $1-\zeta\sim e^{-2\ell\nu}$ with
an integer $\ell$, which is the only other possibility within this class,
is taken up in Section~\ref{sec_scaling}. Once $\zeta(\nu)$ is fixed,
the potential is forced as well: $I=I_h$ means
$Q=\zeta'^2I_h(\zeta)+\frac12\{\zeta,\nu\}$, where the $n$-dependent
part $-n^2$ of $Q$ is of the form $-n^2\mathcal R(\zeta)/(4\zeta^2(1-\zeta)^2)$
and is absorbed in the parameters of $I_h$, while the endpoint
exponents are prescribed by the axis and the circle; what remains
determines $w$ up to the single constant of $I_h$ that is left free,
denoted $S$ below (Section~\ref{sec_modes}). The modes are recovered
from $Y$ through $g=(\zeta')^{-1/2}Y$ and $f_n=g_n/\sqrt{\sinh\nu}$.

The constant $S$ can be negative; it is fixed below, as a function of
$k$, by the boundary condition at the circle. The reduced potential
forced in this way is
\begin{equation}\label{wformula}
 w=\frac{S(1-\sigma^2)}{R^2}
 +\frac{(1-\sigma^2)\bigl[(1-k)^2(1-\sigma^2)+4k\sigma^4\bigr]}
        {4\sigma^2R^6}
 -\frac1{4\sinh^2\nu(\sigma)}.
\end{equation}
The apparent double poles at $\sigma=0$ cancel. The first integration
in \eqref{qfromw} has an elementary primitive,
\begin{align}
 q'(\nu)&=\coth\nu-
          \frac{P(\sigma)}{3\sigma(1-k)R^3}+4S A(\sigma),\label{qprime}\\
 P(\sigma)&=3(1-k)^2+3(1+2k-3k^2)\sigma^2
                       +2k(3k-1)\sigma^4,\label{P}\\
 A(\sigma)&=\int_0^\sigma\frac{\dd t}{R(t)}
 =\frac1{\sqrt{k}}\atanh\frac{\sqrt{k}\,\sigma}{R}.\label{Adef}
\end{align}
For negative $k$, the last expression is
$\arctan(\sqrt{|k|}\,\sigma/R)/\sqrt{|k|}$.
Condition \eqref{balance} is $q'(\infty)=0$ in \eqref{qprime}; with
$\coth\nu\to1$, $P(1)=2(3-k)$ and $A(1)=\atanh\sqrt k/\sqrt k$ at the
circle, it reads $1-2(3-k)/[3(1-k)]+4S\atanh\sqrt k/\sqrt k=0$ and fixes
\begin{equation}\label{Sdef}
 \boxed{S(k)=\frac{(3+k)\sqrt{k}}
 {12(1-k)\atanh\sqrt{k}}},\qquad S(0)=\frac14.
\end{equation}
The real continuation for $k<0$ uses
$\atanh\sqrt{k}/\sqrt{k}=\arctan\sqrt{|k|}/\sqrt{|k|}$.
The metric function is normalized by $q(0)=0$. At the axis,
\begin{equation}\label{waxis}
 w(0)=\frac{k^2-6k-3+12(1-k)S(k)}{12(1-k)^2}.
\end{equation}

\subsection{An explicit dilogarithmic metric function}
The second integration can also be performed explicitly. For $k\ne0$, put
\begin{equation}\label{chitheta}
 c=\frac{1-\sqrt k}{1+\sqrt k},\qquad t=\frac{R-\sigma}{R+\sigma},
\end{equation}
and define the following finite combination of dilogarithms,
$\operatorname{Li}_2(x)=-\int_0^x\log(1-s)\,\dd s/s$:
\begin{equation}\label{Ldef}
 \mathcal L=\operatorname{Li}_2(-ct)-\operatorname{Li}_2(-c)
           -\operatorname{Li}_2(-t/c)+\operatorname{Li}_2(-1/c).
\end{equation}
Then the normalized metric function is
\begin{equation}\label{qclosed}
 \boxed{
 \begin{aligned}
 q(\sigma;k)={}&\log\frac{\sinh\nu}{\sqrt{1-k}\,\sigma}
   -\frac56\log\frac{R^2}{1-k}
   +\frac{3-k}{3(1-k)}\log(1-\sigma^2)\\
   &+\frac{S}{\sqrt k}\Bigl[4\,\atanh\frac{\sigma}{R}\,\atanh\sqrt k
     +\mathcal L\Bigr]-2S\,\atanh^2\frac{\sqrt k\,\sigma}{R} .
 \end{aligned}}
\end{equation}
The formula holds for $0<\sigma<1$ and extends to the endpoints by limits.
At $k=0$ one has $q=0$. For $0<k<1$ all arguments of the dilogarithms are
negative and \eqref{qclosed} is manifestly real. For $k<0$, choose
$\sqrt k=i\sqrt{|k|}$ and the principal branches. Then $c$ is
unimodular, $-1/c=-\bar c$, $t$ is real, so the paired dilogarithms in
\eqref{Ldef} are complex conjugates, while $\atanh\sqrt k$ and
$\atanh(\sqrt k\,\sigma/R)$ are purely imaginary. Consequently
\eqref{qclosed} is real. Equations \eqref{nusigma} and \eqref{qclosed}
give an explicit parametric representation; no quadrature is required,
although evaluating $q$ at a prescribed $\nu$ generally requires
inversion of \eqref{nusigma}.

To derive the expression, put $b=\sqrt k$, $\chi=\atanh(\sigma/R)$ and
$B=\atanh(b\sigma/R)$, so that $\nu=\chi-bB$ by \eqref{nusigma} and
$t=e^{-2\chi}$, and set $C=(3-k)/[3(1-k)]$. The rational part of
$q'(\nu)\,\dd\nu$ has the elementary primitive
\begin{equation}\label{rationalprimitive}
 \frac{P(\sigma)}{3\sigma(1-k)R^2(1-\sigma^2)}
 =\frac{\dd}{\dd\sigma}\left[\log\sigma+\frac56\log R^2
        -C\log(1-\sigma^2)\right].
\end{equation}
For the remaining term, $A=B/b$ and $\nu=\chi-bB$ give
\begin{equation}
 4S\int_0^\nu A\,\dd\nu
 =\frac{4S}{b}\int_0^\chi\atanh(b\tanh t)\,\dd t-2SB^2 .
\end{equation}
Now use
\begin{align}
 \atanh(b\tanh\chi)&=\atanh b
   +\frac12\log\frac{1+ce^{-2\chi}}{1+c^{-1}e^{-2\chi}},\label{thetaidentity}\\
 \frac{\dd}{\dd\chi}\operatorname{Li}_2\bigl(-ce^{-2\chi}\bigr)
   &=2\log\bigl(1+ce^{-2\chi}\bigr).
\end{align}
Integrating from zero yields \eqref{qclosed}; its normalization follows
from $\nu\sim\sqrt{1-k}\,\sigma$ and from $\mathcal L=B=0$ at $\sigma=0$,
where $t=1$.

For completeness, let $L=\atanh\sqrt k$. Cancellation of the terms linear in
$\chi$, which is condition \eqref{Sdef}, gives the finite circle value
\begin{equation}\label{qinf}
 \begin{split}
 q_\infty={}&-\sqrt k\,L+\frac{3+k}{3(1-k)}\log2
            -\frac{2}{3(1-k)}\log(1-k)-2SL^2\\
 &+\frac{S}{\sqrt k}\Bigl[\operatorname{Li}_2(-1/c)
                 -\operatorname{Li}_2(-c)\Bigr].
 \end{split}
\end{equation}

At $k=-3$ the coefficient $S$ vanishes. This non-flat member lies in the
physical existence interval established below and has an elementary
profile:
\begin{equation}\label{qminus3}
 \boxed{q(\sigma;-3)=\log\frac{\sinh\nu}{2\sigma}
       -\frac56\log\frac{4-3\sigma^2}{4}+\frac12\log(1-\sigma^2).}
\end{equation}
Here
\begin{equation}
 \nu=\atanh\frac{\sigma}{\sqrt{4-3\sigma^2}}
     +\sqrt3\arctan\frac{\sqrt3\,\sigma}{\sqrt{4-3\sigma^2}},\qquad
 q_\infty=\frac{\pi}{\sqrt3}-\frac{\log2}{3}.
\end{equation}
In particular $w(0)=1/8$ and $m/a=0.23114878$, so this is not the flat
$k=0$ solution.

\begin{proposition}\label{regularity}
For each fixed $k<1$, $q$ and $V$ extend to real-analytic functions on
$\R^3$. The seed metric in \eqref{metric} is regular there, $q$ is bounded,
the boundary conditions \eqref{bcq} are satisfied,
and \eqref{q_asym}--\eqref{V_asym} hold.
\end{proposition}
\begin{proof}
At the axis $\sigma$ is an odd analytic function of $\nu$ and
\eqref{wformula} has an even analytic extension. Integrating with
$q(0)=q'(0)=0$ gives $q=O(\nu^2)$. In physical coordinates,
$\tanh\nu=2a\rho/(r^2+a^2)$ by \eqref{toroidalrelations}, so $q$ is analytic in $\rho^2,z$ and
$e^{2q}-1=O(\rho^2)$, precisely the regular-axis condition for the metric.

At the circle put $\varepsilon=1-\sigma$ and $x=e^{-2\nu}$. Equation
\eqref{Rdef} gives
$x\,\dd\varepsilon/\dd x=\varepsilon(2-\varepsilon)/(2R)$.
Equivalently, separation yields
$\log x=\log\varepsilon+\text{an analytic function of }\varepsilon
+\text{constant}$. The inverse function theorem therefore gives
$\varepsilon=Cx+O(x^2)$ analytically, with $C>0$. Formula
\eqref{wformula} is analytic in $x$ and vanishes at $x=0$. Since
$q'(\infty)=0$, both $q'$ and $q-q_\infty$ are analytic in $x$ and
$O(x)$. Equation \eqref{x_circle} proves physical analyticity there;
the factor $d^2w$ in $V$ also has a regular analytic extension. Away from
these coordinate loci the claim is immediate. Expanding at $\nu=0$
and using \eqref{coords} proves the asymptotic statements.
\end{proof}

\begin{figure}[tb]
\centering\includegraphics[width=\linewidth]{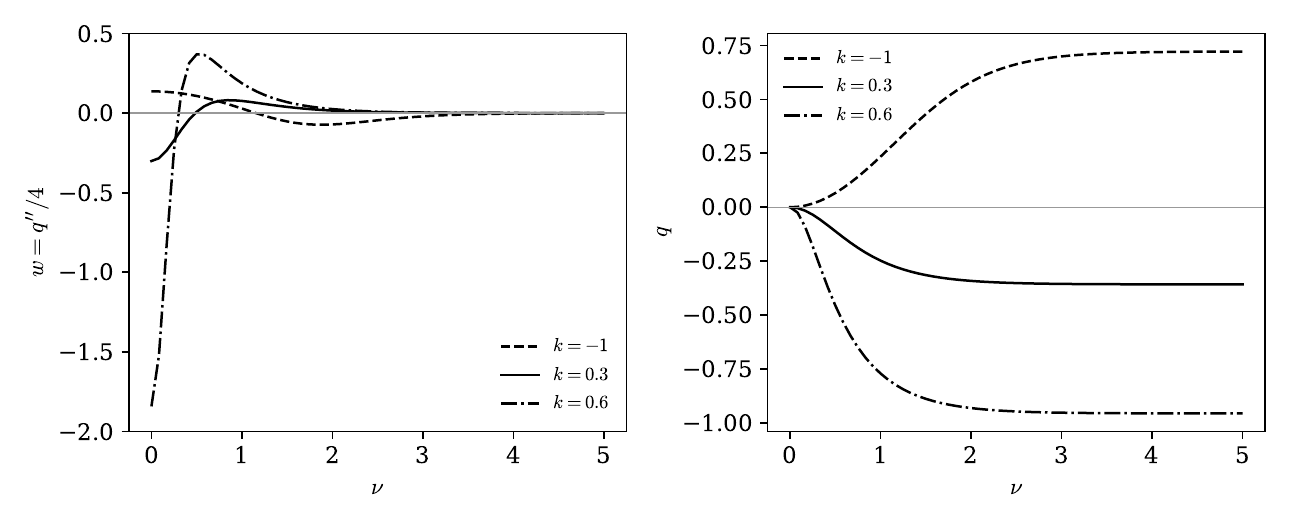}
\caption{Analytic profiles for $k=-1,0.3,0.6$. The coordinate $\nu=0$
is the axis and $\nu\to\infty$ approaches the circle. The limiting constant
$q_\infty$ is compatible with regularity.}\label{fig_profiles}
\end{figure}

\section{A ground-state identity and the exact existence interval}
Let $\D$ be the completion of $C_c^\infty(\R^3)$ in the norm
$\|\nabla\eta\|_{L^2}$, and define
\begin{equation}\label{energy}
 E_k[\eta]=\int_{\R^3}(|\nabla\eta|^2-V\eta^2)\,\dd^3x.
\end{equation}
All quantities in this quadratic form refer to the Euclidean metric.
The potential is bounded and belongs to $L^{3/2}\cap L^{6/5}$.
The potential term is a compact quadratic form on $\D$: approximate $V$
in $L^{3/2}$ by bounded compactly supported functions and use local Rellich
compactness and the Sobolev inequality.

\begin{lemma}\label{auxlemma}
For every $k<1$, the positive function
\begin{equation}\label{auxh}
 h_k=\sqrt d\,\sqrt{\frac{\sigma R(\sigma)}{\sinh\nu}}
\end{equation}
extends smoothly to $\R^3$, belongs to $\D$, and satisfies
\begin{equation}\label{auxeq}
 (\Delta_3+V)h_k=(S-1)C_kh_k,\qquad
 C_k=\frac{d^2(1-\sigma^2)}{a^2R^2}>0.
\end{equation}
Here $C_k$ has a smooth positive extension across the coordinate loci.
Furthermore,
\begin{align}
 E_k[\eta]&=\int h_k^2\left|\nabla(\eta/h_k)\right|^2\dd^3x
           +(1-S)\int C_k\eta^2\dd^3x,
               &&\eta\in C_c^\infty(\R^3),\label{groundstate}\\
 E_k[h_k]&=2\pi^2a(1-S).\label{trialenergy}
\end{align}
\end{lemma}
\begin{proof}
For the radial factor $f=\sqrt{\sigma R/\sinh\nu}$, direct substitution
in \eqref{modeeq}, with $n=0$, gives the residual
$(S-1)(1-\sigma^2)f/R^2$. This proves \eqref{auxeq}. The endpoint
expansions used in Proposition~\ref{regularity} prove smoothness and
positivity. At the circle $h_k\to1$, and at infinity
$h_k=\sqrt2a/r+O(r^{-3})$, so $h_k\in\D$.
Expanding $\nabla(h_k\eta/h_k)$ and integrating by parts proves
\eqref{groundstate}. There is no artificial boundary at the circle in
this calculation. Finally, the Euclidean volume element is
$a^3\sinh\nu\,d^{-3}\dd\nu\dd u\dd\varphi$, whence
\begin{equation}
 \int C_kh_k^2\dd^3x
 =4\pi^2a\int_0^1\sigma\,\dd\sigma=2\pi^2a.
\end{equation}
Integration by parts at infinity has a vanishing boundary term for $h_k$,
and gives \eqref{trialenergy}.
\end{proof}

The lemma has a simple reading in the language of \eqref{constraints}.
By \eqref{wformula}, $V=S\,C_k$ plus terms independent of $S$, so the
constant $S$ measures the depth of the potential well $-V$ in the
Schr\"odinger equation for $\Phi$, and $S=1$ is exactly the threshold at
which the well begins to bind: by \eqref{groundstate} there is no bound
state for $S<1$, by \eqref{auxeq} the function $h_k$ is a bound state at
zero energy for $S=1$, and by \eqref{trialenergy} the energy is negative
on $h_k$ for $S>1$.

\begin{lemma}\label{critical}
The equation $S(k)=1$ has exactly one solution $k_c<1$. It lies in $(0,1)$,
with $S<1$ for $k<k_c$ and $S>1$ for $k_c<k<1$. Explicitly,
\begin{equation}\label{kceq}
 12(1-k_c)\frac{\atanh\sqrt{k_c}}{\sqrt{k_c}}=3+k_c,
 \qquad k_c=0.80480167722084357844\ldots .
\end{equation}
The quotient in this equation is continuously extended at $k=0$.
\end{lemma}
\begin{proof}
For $0<k<1$,
\begin{equation}
 (1-k)\frac{\atanh\sqrt k}{\sqrt k}
       =1-\sum_{j=1}^\infty\frac{2k^j}{4j^2-1}
\end{equation}
decreases strictly from $1$ to $0$, while $3+k$ increases. The two
sides of \eqref{kceq} cross once. For $-3<k\le0$, the inequality
$\arctan t\ge t/(1+t^2)$ gives $S\le(3+k)/12<1$.
For $k\le-3$ one has $S\le0$.
\end{proof}

\begin{theorem}\label{existence}
For every $a>0$ and $k<k_c$ there is a unique solution
$\Phi=1+\phi$, $\phi\in\D$, of \eqref{constraints}. It is positive,
real analytic on $\R^3$, axisymmetric and even in $z$. It has the asymptotic
form
\begin{equation}\label{AF}
 \Phi=1+\frac{m}{2r}+O_2(r^{-2}).
\end{equation}
The metric \eqref{metric} is complete, has one asymptotically flat end,
and satisfies the vacuum constraint everywhere. No positive solution
with the asymptotics \eqref{AF} exists for $k_c\le k<1$.
\end{theorem}
\begin{proof}
For $S<1$, \eqref{groundstate}, extended by approximation, implies
$E_k[\eta]\ge(1-S)\int C_k\eta^2>0$ for nonzero $\eta\in\D$.
Because $E_k$ is the identity minus a compact self-adjoint form in the
Dirichlet inner product, this strict positivity is coercivity:
$E_k[\eta]\ge c_k\|\nabla\eta\|_2^2$, with $c_k>0$.
Indeed, failure of coercivity would make zero an eigenvalue of that compact
perturbation, contradicting strict positivity.

The functional $\eta\mapsto\int V\eta$ is bounded on $\D$ by
$V\in L^{6/5}$ and Sobolev embedding. Lax--Milgram therefore gives the
unique solution of
$(-\Delta_3-V)\phi=V$. Interior elliptic regularity and the analytic
coefficients give local analyticity. The standard asymptotically flat
elliptic estimates, or the Newton-potential argument in
Appendix~\ref{app_asym}, give \eqref{AF}. This is the usual variational
framework for the scalar-curvature equation \cite{CantorBrill}.

Since $\Phi\to1$, its negative part has compact support. Testing the
homogeneous equation for $\Phi$ with this negative part gives zero energy,
so coercivity implies $\Phi\ge0$. Harnack's inequality gives $\Phi>0$.
Uniqueness yields the axial and equatorial symmetries. Positivity, continuity
and the limit at infinity bound $\Phi$ above and away from zero for fixed
$k$. Together with bounded $q$, this makes $g$ uniformly equivalent to the
Euclidean metric and proves completeness. The regular-axis condition was
proved above.

At $S=1$, $h_k$ is a positive zero mode. Green's identity with a putative
$\Phi\to1$ gives
\begin{equation}
 0=\lim_{r\to\infty}\int_{S_r}
       (\Phi\partial_r h_k-h_k\partial_r\Phi)\,\dd S
    =-4\pi\sqrt2a,
\end{equation}
a contradiction. For $S>1$, \eqref{trialenergy} is negative. A positive
solution $\Phi$ would instead give
$E_k[\eta]=\int\Phi^2|\nabla(\eta/\Phi)|^2\ge0$ for every compactly
supported test function. Approximating $h_k$ in $\D$ contradicts its
negative energy. Lemma~\ref{critical} finishes the proof.
\end{proof}

For other seed functions, lower and upper limits on the amplitude for
which the constraint has a solution were determined numerically in
\cite{MastersonHobill}. In the present family there is no lower limit,
and the upper limit $k_c$ is exact.

\section{Explicit modes and identification of the global solution}\label{sec_modes}
The invariant $I(\zeta)$ in \eqref{normalform} resulting from \eqref{natanzon} is
\begin{equation}
 I(\zeta)=\frac{(B_0-kn^2/4)\zeta^2-
              (B_0+(1-k)n^2/4)\zeta+1/4}
             {\zeta^2(1-\zeta)^2},\qquad B_0=\frac{1-S}{4}.
\end{equation}
Explicitly, with $x=1-\sigma^2$ and
$f_n=\sqrt{\sigma R/\sinh\nu}\,(1-\sigma^2)^{n/2}\,y_n(x)$, the mode
equation \eqref{modeeq} becomes, up to the factor $4(1-\sigma^2)/R^2$,
the Gauss hypergeometric equation
\begin{equation}\label{hypergeom}
 x(1-x)\,y_n''+\bigl[(n+1)-(n+2)x\bigr]y_n'-\alpha_n\beta_n\,y_n=0,
\end{equation}
whose parameters $a=\alpha_n$, $b=\beta_n$, $c=n+1$ satisfy
$\alpha_n+\beta_n=n+1$ and $\alpha_n\beta_n=\tfrac14[(n+1)^2-S-kn^2]$.
Solving these two relations, define
\begin{equation}\label{ab}
 s_n=\sqrt{S+kn^2},\qquad
 \alpha_n=\frac{1+n-s_n}{2},\qquad
 \beta_n=\frac{1+n+s_n}{2},\qquad
 \epsn=\begin{cases}1,&n=0,\\2,&n\ge1.\end{cases}
\end{equation}
The sign of the square root does not matter because it interchanges
$\alpha_n$ and $\beta_n$. For an imaginary $s_n$ these parameters are
complex conjugates. On the physical branch, $k<k_c$, neither parameter is a
non-positive integer: if $s_n$ is real, then $S<1$ and $k<1$ give
$s_n^2<1+n^2\le(n+1)^2$, hence $\alpha_n>0$; otherwise both parameters
are non-real. The circle-regular modes, with their normalizations,
are
\begin{equation}\label{explicit_modes}
 \boxed{
 f_n=\frac{\epsn}{\sqrt2\pi}
       \frac{\Gamma(\alpha_n)\Gamma(\beta_n)}{\Gamma(n+1)}
       \sqrt{\frac{\sigma R}{\sinh\nu}}
       (1-\sigma^2)^{n/2}
       {}_2F_1(\alpha_n,\beta_n;n+1;1-\sigma^2).}
\end{equation}
These expressions are real. They behave as $e^{-(n+1/2)\nu}$ at the
circle. For $n=0$, the second independent solution has an additional
logarithm there and is excluded by regularity.

Since $\alpha_n+\beta_n=n+1$, the zero-balanced connection formula
\cite{AS,DLMF} gives
\begin{equation}\label{lognorm}
 f_n=\frac{\sqrt2\epsn}{\pi}
 \left[-\log\nu+\frac12\log(1-k)-\gamma
          -\frac{\psi(\alpha_n)+\psi(\beta_n)}2+o(1)\right]
 \quad(\nu\to0).
\end{equation}
Here $\psi$ denotes the digamma function and $\gamma$ is Euler's constant.
Compare this with Heine's identity
\cite{Cohl}
\begin{equation}\label{Heine}
 \frac1{\sqrt d}
  =\frac{\sqrt2}{\pi}\sum_{n=0}^\infty
       \epsn Q_{n-1/2}(\cosh\nu)\cos(nu),
\end{equation}
and
$Q_{n-1/2}(\cosh\nu)=-\log(\nu/2)-\gamma-\psi(n+1/2)+o(1)$.

\begin{theorem}\label{reconstruction}
For $k<k_c$, the unique solution of Theorem~\ref{existence} is
\begin{equation}\label{Phiseries}
 \boxed{\Phi=\sqrt d\sum_{n=0}^\infty f_n(\nu)\cos(nu),}
\end{equation}
with \eqref{explicit_modes}. On compact subsets with $\nu>0$ and
finite $\nu$, the series and each fixed number of derivatives converge
uniformly. Define
\begin{equation}\label{correctedseries}
 D=\frac{\Phi-1}{\sqrt d},\qquad
 d_n=f_n-\frac{\sqrt2\epsn}{\pi}Q_{n-1/2}(\cosh\nu).
\end{equation}
Then $D$ extends continuously to $\nu=0$, including $u=0$. Its axis
Fourier series, with coefficients $d_n(0)$, converges absolutely and
uniformly, and its value at $(0,0)$ determines the ADM mass.
\end{theorem}
\begin{proof}
Start with the solution already obtained in Theorem~\ref{existence}.
For fixed $\nu>0$, its function $F=\Phi/\sqrt d$ is smooth and periodic
in $u$, and is even. Its cosine coefficients satisfy \eqref{modeeq}.
Smoothness in the transverse coordinates at the circle forces the
$n$th coefficient to be $O(e^{-(n+1/2)\nu})$. It is therefore a multiple
of the regular hypergeometric branch in \eqref{explicit_modes}.

Away from $(0,0)$, regularity of $\Phi$ implies a continuous limit for $D$
at $\nu=0$. Near $(0,0)$, \eqref{AF} and
$\sqrt d=\sqrt2a/r+O(r^{-3})$ imply
\begin{equation}\label{Dlimit}
 D=\frac{m}{2\sqrt2a}+O\bigl(\sqrt{\nu^2+u^2}\bigr).
\end{equation}
Thus $D$ is bounded and continuous on a closed strip containing the
axis, with the two copies of $u=\pi$ identified. Its Fourier coefficients
are bounded as $\nu\to0$. The coefficient of $-\log\nu$ in every
Fourier coefficient of $F$ must consequently equal the one in
\eqref{Heine}. Formula \eqref{lognorm} fixes the multiplicative constant
uniquely and gives \eqref{explicit_modes}.

Fourier reconstruction of the smooth periodic function $F$ now proves
\eqref{Phiseries}; repeated integration by parts in $u$ proves the
stated compact-subset convergence, including derivatives. No global
solution is inferred from a formal mode sum.

Continuity of $D$ on the closed strip allows passage to the limit under
its Fourier integrals. Subtracting the two finite parts gives
\begin{equation}\label{dnaxis}
 d_n(0)=\frac{\sqrt2\epsn}{\pi}
 \left[\psi\left(n+\frac12\right)
 -\frac{\psi(\alpha_n)+\psi(\beta_n)}2
 +\frac12\log\frac{1-k}{4}\right].
\end{equation}
Denote the bracket in \eqref{dnaxis} by $T_n$. Only symmetric functions
of $\alpha_n$ and $\beta_n$ enter it, and
\begin{equation}\label{symmetric}
 \alpha_n+\beta_n=n+1,\qquad
 \alpha_n\beta_n=\frac{\Delta_n}{4},\qquad
 \Delta_n=(1-k)n^2+2n+1-S .
\end{equation}
For $k\ne0$ both parameters tend to infinity inside the sector
$|\arg x|<\pi$, where
$\psi(x)=\log x-1/(2x)-1/(12x^2)+1/(120x^4)+O(x^{-6})$. Hence
\begin{equation}
 \begin{split}
 T_n={}&\log\Bigl(n+\tfrac12\Bigr)-\frac12\log\frac{\Delta_n}{1-k}
       -\frac1{2n+1}-\frac1{3(2n+1)^2}+\frac{2}{15(2n+1)^4}\\
       &+\frac{n+1}{\Delta_n}
        +\frac{2(n+1)^2-\Delta_n}{3\Delta_n^2}
        -\frac{4\bigl[\bigl(2(n+1)^2-\Delta_n\bigr)^2-\Delta_n^2/2\bigr]}
              {15\Delta_n^4}+O_k(n^{-6}).
 \end{split}
\end{equation}
Expanding in powers of $1/n$, the terms of order $1$, $n^{-1}$ and
$n^{-3}$ cancel identically, the coefficient of $n^{-2}$ is $w(0)/2$ with
$w(0)$ given by \eqref{waxis}, and
\begin{equation}\label{tail}
 T_n=\frac{w(0)}{2n^2}+\frac{t_4}{n^4}+O_k(n^{-6}),\qquad
 d_n(0)=\frac{\sqrt2\epsn}{\pi}\,T_n ,
\end{equation}
\begin{equation}\label{b4}
 t_4=\frac{240S^2(1-k)^2+160S(k+3)(k-1)-7k^4+28k^3+70k^2+700k+105}
          {960(1-k)^4}.
\end{equation}
For $k=0$ one has $T_n\equiv0$, consistently with $w(0)=t_4=0$.
The resulting axis series is absolutely and uniformly convergent.
It has the Fourier coefficients of the continuous function $D(0,u)$,
so Fourier uniqueness identifies it with that function, including at
$u=0$. This proves the mass evaluation without assuming uniform
convergence of the original singular series near the point at infinity.
\end{proof}

\paragraph{The explicit solution in one place.}
With Theorem~\ref{reconstruction} the solution is complete. For every
$k<k_c$ the data are the metric \eqref{metric}, in which the seed
function $q$ is \eqref{qclosed}, with the variable $\sigma(\nu)$ defined
by \eqref{Rdef}--\eqref{nusigma} and the constant $S(k)$ by \eqref{Sdef},
and the conformal factor $\Phi$ is the series \eqref{Phiseries} with the
modes \eqref{explicit_modes} and their parameters \eqref{ab}. These
formulas are the solution: no thin-torus or weak-field approximation is
involved.

At the axis the original series \eqref{Phiseries} should not be evaluated
term by term: individual modes have logarithmic singularities. One uses
\eqref{correctedseries} and \eqref{dnaxis}. At the circle one uses the
smooth limiting extension. For $k=0$, uniqueness gives $\Phi\equiv1$
and $q\equiv0$; the modes reduce to the coefficients in \eqref{Heine}.

On the circle itself the conformal factor has a closed form. The
coefficient $f_0$ is the average of $\Phi/\sqrt d$ over $u$, and $\Phi$
is continuous at the circle, so its value there is the limit of
$\sqrt d\,f_0$ as $\nu\to\infty$. Since $d\,\sigma R/\sinh\nu\to1$,
\eqref{explicit_modes} and the reflection formula for the Gamma function
give
\begin{equation}\label{Phicircle}
 \Phi\big|_{\rho=a,\,z=0}
 =\frac{\Gamma(\alpha_0)\Gamma(\beta_0)}{\sqrt2\,\pi}
 =\frac1{\sqrt2\,\cos\bigl(\pi\sqrt S/2\bigr)} ,
\end{equation}
with $\cos(\pi\sqrt S/2)=\cosh(\pi\sqrt{|S|}/2)$ for $S<0$. This equals
$1$ at $k=0$, exceeds $1$ precisely for $k>0$, equals $1/\sqrt2$ for the
elementary member $k=-3$, and diverges as $k\uparrow k_c$.
Figure~\ref{figfields} shows $q$ and $\Phi$ in a meridional plane for
two members of the family.

\begin{figure}[tb]
\centering\includegraphics[width=\linewidth]{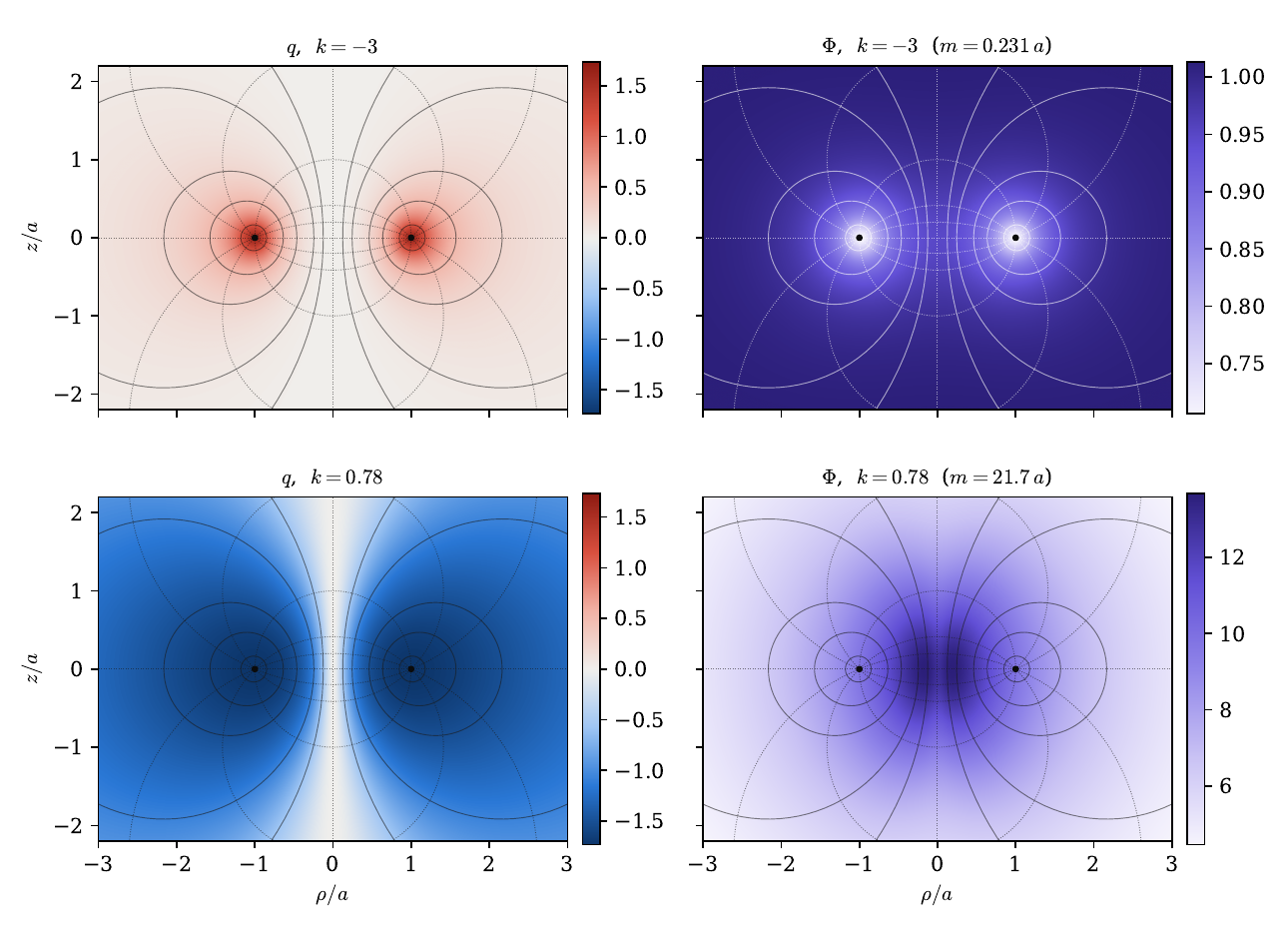}
\caption{The metric function $q$ (left) and the conformal factor $\Phi$
(right) in a meridional plane, with $\rho$ continued to negative values
across the axis, for the elementary member $k=-3$ (top) and for
$k=0.78$ (bottom), which contains trapped surfaces. Thin solid circles
are sections of the tori $\nu=0.25,0.5,1,1.5,2.5$, dotted circles are
surfaces $u=\mathrm{const}$, and the dots mark the circle $\rho=a$,
$z=0$. The function $q$ is constant on the tori and has opposite signs
on the two branches of the family; both rows use the same colour scale
for $q$, while each panel for $\Phi$ has its own. On the circle $\Phi$
takes the values \eqref{Phicircle}, $1/\sqrt2$ and $9.752$. For $k=0.78$
the outer marginally trapped sphere of Table~\ref{tablehorizons} lies
outside the frame.}\label{figfields}
\end{figure}

\section{ADM mass and its critical divergence}
The seed function decays as $r^{-2}$, so the ADM mass is the coefficient
$m$ in \eqref{AF}. Theorem~\ref{reconstruction} gives the explicit formula
\begin{equation}\label{mass}
 \boxed{\frac ma=\frac4\pi\sum_{n=0}^\infty\epsn
 \left[\psi\left(n+\frac12\right)
 -\frac{\psi(\alpha_n)+\psi(\beta_n)}2
 +\frac12\log\frac{1-k}{4}\right].}
\end{equation}
The series converges absolutely by \eqref{tail}. At $k=0$, each term
vanishes by the duplication formula for the digamma function.
A numerical truncation at $n=N-1$ can be accelerated by adding the known
tails from \eqref{tail}:
\begin{equation}\label{masstail}
 \frac ma=\frac4\pi\left[
    \sum_{n=0}^{N-1}\epsn T_n+w(0)\zeta(2,N)+2t_4\zeta(4,N)\right]
    +O_k(N^{-5}),
\end{equation}
where $T_n$ is the bracket in \eqref{mass}, $t_4$ is given by \eqref{b4}
and $\zeta(s,N)$ is the Hurwitz zeta function. With $N=500$ and $N=2000$
the results agree to eleven digits for the parameters of
Table~\ref{tablemass}. At very small $k$, cancellation calls for
higher precision rather than merely a larger $N$.

At the critical point, $\alpha_0\to0^+$ and all $n\ge1$ terms remain
regular. Since $\alpha_0=(1-\sqrt S)/2$ and
$\psi(\alpha_0)=-1/\alpha_0-\gamma+O(\alpha_0)$,
\begin{equation}\label{masspole}
 \frac ma=\frac8{\pi(1-S(k))}+O(1)
 =\frac8{\pi S'(k_c)}\frac1{k_c-k}+O(1)
 \quad(k\uparrow k_c).
\end{equation}
The $O(1)$ term requires uniform control of the sum over $n\ge1$, not
only regularity of each of its terms. Let $K\subset(0,1)$ be a closed
interval containing $k_c$ in its interior, on which $S\le2$. For $n\ge1$
and $k\in K$, \eqref{symmetric} gives
$\alpha_n\beta_n=\Delta_n/4\ge(1-k)n^2/4$ and $\beta_n\le n+2$, hence
$\beta_n\ge\alpha_n\ge c_Kn$ with $c_K>0$. The digamma expansion leading
to \eqref{tail} therefore has a remainder that is uniform on $K$, its
coefficients $w(0)$ and $t_4$ are continuous there, and
\begin{equation}\label{uniformtail}
 |T_n(k)|\le\frac{C_K}{n^2}\qquad(n\ge1,\ k\in K).
\end{equation}
Consequently $\sum_{n\ge1}\epsn T_n(k)$ converges uniformly on $K$ and is
continuous there; only the term $n=0$ is singular at $k_c$.

In accordance with \cite{BeigOM}, trapped surfaces appear before the
critical point is reached: this is proved in Section~\ref{sec_proof} and
examined numerically in Section~\ref{sec_trapped}.
This realizes explicitly the mass blow-up associated with loss of
positivity of the conformal operator \cite{BeigOM,BeigHusa}; the general
phenomenon is not new. For other seed functions the dependence of the
mass on the amplitude is known numerically, with fitted scaling
relations \cite{AranhaOliveira}. Here \eqref{mass} gives it exactly on
the whole interval of existence, from the quadratic regime
\eqref{weakmass} to the pole \eqref{masspole}.

\begin{table}[tb]
\centering
\begin{tabular}{rrrrr}
\toprule
$k$&$S(k)$&$w(0)$&$q_\infty$&$m/a$\\\midrule
$-50$&$-0.37967$&$0.08217$&$8.3646$&$1.02889842$\\
$-10$&$-0.13262$&$0.09607$&$3.3591$&$0.51552735$\\
$-3$&$0$&$0.12500$&$1.5828$&$0.23114878$\\
$-1$&$0.10610$&$0.13638$&$0.7219$&$0.08375681$\\
$-0.1$&$0.22683$&$0.04161$&$0.0955$&$0.00259299$\\
$0.1$&$0.27720$&$-0.06134$&$-0.1053$&$0.00390123$\\
$0.3$&$0.34981$&$-0.30129$&$-0.3575$&$0.06013333$\\
$0.5$&$0.46800$&$-0.98068$&$-0.7077$&$0.36598300$\\
$0.7$&$0.71070$&$-3.84396$&$-1.3046$&$3.20159994$\\
$0.78$&$0.90973$&$-8.04048$&$-1.7341$&$21.66364686$\\
\bottomrule
\end{tabular}
\caption{Representative profiles and masses. The values of $q_\infty$ follow
from \eqref{qinf}; the mass is evaluated from \eqref{mass}. The length scale $a$
does not affect these dimensionless quantities.}\label{tablemass}
\end{table}

\begin{figure}[tb]
\centering\includegraphics[width=\linewidth]{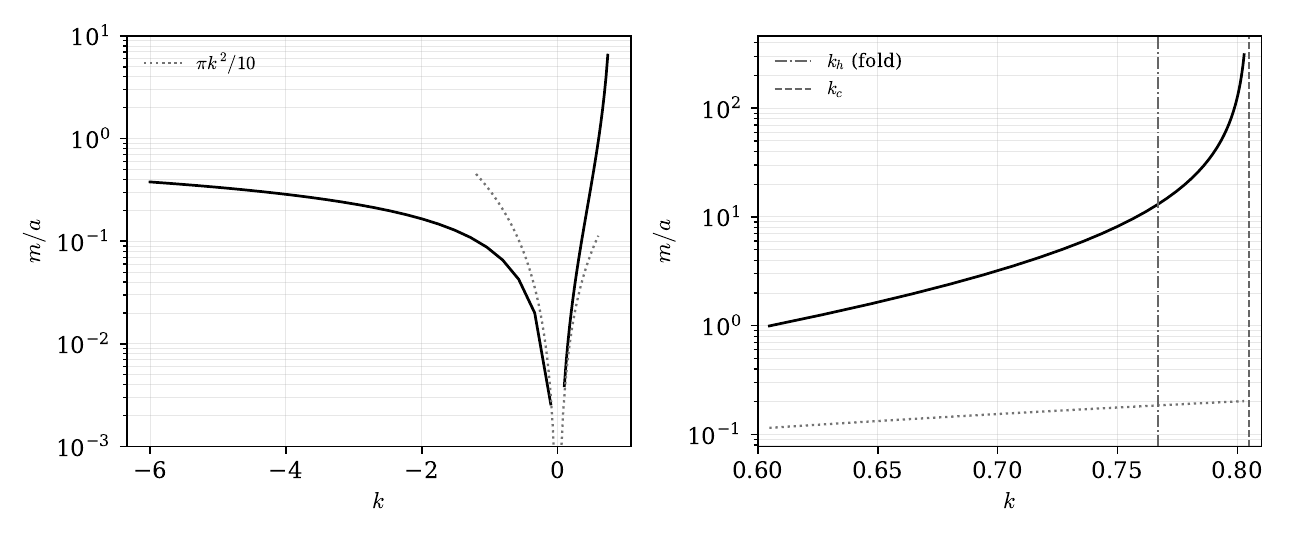}
\caption{The ADM mass of the physical branch. Left: both signs of $k$.
Right: approach to the critical point. The dotted curve is $\pi k^2/10$;
the vertical lines indicate the numerical fold of marginally trapped spheres and the
exact existence boundary.}\label{figmass}
\end{figure}

\section{An analytic small-amplitude expansion}
Expanding \eqref{Sdef} gives
$S=1/4+k/4+k^2/5+O(k^3)$. At fixed physical coordinates,
\begin{equation}\label{qweak}
 q=kq_1+O(k^2),\qquad
 q_1=-\tanh^2\nu=-\frac{4a^2\rho^2}{(r^2+a^2)^2}.
\end{equation}
Write $\Phi=1+k\phi_1+O(k^2)$, so that
$\Delta_3\phi_1=-\Delta_2q_1/4$. Put $\xi=r/a$ and
$P_2(\mu)=(3\mu^2-1)/2$, and let $\theta$ be the polar angle,
$\rho=r\sin\theta$, $z=r\cos\theta$. The unique regular decaying solution is
\begin{align}
 \phi_1&=U_0(\xi)+U_2(\xi)P_2(\cos\theta),\label{phiweak}\\
 U_0(\xi)&=\frac{\xi^2+1/3}{(1+\xi^2)^2},\\
 U_2(\xi)&=\frac{1+5\xi^2/3}{\xi^2(1+\xi^2)^2}
                         -\frac{\arctan\xi}{\xi^3}.
\end{align}
The apparent poles at $\xi=0$ cancel:
$U_0(0)=1/3$ and $U_2=-8\xi^2/15+O(\xi^4)$.
For verification, the monopole and quadrupole parts of
$a^2\Delta_2q_1/4$ are
\begin{equation}\label{Vweak}
 v_0=\frac{-2\xi^4+12\xi^2-2}{(1+\xi^2)^4},\qquad
 v_2=-\frac{16\xi^2}{(1+\xi^2)^4}.
\end{equation}
Substitution in the radial Poisson equations proves \eqref{phiweak}.

By \eqref{zeromean}, $\int V\,\dd^3x=0$.
Integrating the equation for $\log\Phi$ yields Brill's identity in this
normalization,
\begin{equation}\label{Brillmass}
 m=\frac1{2\pi}\int_{\R^3}|\nabla\log\Phi|^2\,\dd^3x\ge0.
\end{equation}
Analytic dependence of the coercive problem on $k$ near zero justifies
using the first variation in this integral. Hence
\begin{align}
 \frac ma&=c_2k^2+O(k^3),\\
 c_2&=2\int_0^\infty \xi^2\left(U_0v_0+\frac15U_2v_2\right)\dd\xi.
\end{align}
One integration by parts removes the $\arctan\xi$ term, giving
\begin{equation}
 c_2=\frac4{15}\int_0^\infty
 \frac{4-21\xi^2-21\xi^4+85\xi^6-15\xi^8}{(1+\xi^2)^6}\,\dd\xi
 =\frac\pi{10}.
\end{equation}
The final equality follows termwise from
$\int_0^\infty \xi^{2j}(1+\xi^2)^{-6}\dd\xi
=\frac12 B(j+1/2,11/2-j)$, for $0\le j\le4$. Thus
\begin{equation}\label{weakmass}
 \boxed{\frac ma=\frac\pi{10}k^2+O(k^3).}
\end{equation}
This coefficient is analytic, not inferred by fitting a numerical mass curve.

\section{Trapped surfaces near the critical parameter: an analytic proof}\label{sec_proof}
Let $H$ be the mean curvature of a closed surface in the initial
hypersurface, taken with respect to the outer normal. For time-symmetric
data the expansions of the outgoing and ingoing future-directed null
normals are
\begin{equation}\label{expansions}
 \theta_\pm=\pm H .
\end{equation}
Hence a surface is marginally trapped, $\theta_+=\theta_-=0$, precisely
when it is minimal, and it is outer trapped, $\theta_+<0$, when $H<0$.
Surfaces with both expansions negative cannot lie in a time-symmetric
hypersurface; in this setting ``trapped surfaces'' are therefore understood as the
marginally trapped (minimal) surfaces and the outer trapped surfaces they
enclose. An apparent
horizon is an outermost marginally trapped surface.

The existence of trapped surfaces near the critical amplitude, at
which the conformal operator loses positivity, was established by Beig
and \'O Murchadha \cite{BeigOM}; the present family is the first in
which this mechanism can be followed explicitly, with an exact critical
parameter and trapped surfaces exhibited in closed form. Only the mode $n=0$ of
\eqref{Phiseries} is singular at $k_c$; the rescaled conformal factor
converges to the auxiliary function $h_{k_c}$ of Lemma~\ref{auxlemma};
and the limiting geometry contains explicit surfaces of negative mean
curvature.

Fix $k_1\in(0,k_c)$ and put $J=[k_1,k_c]$. The dependence on $k$ is
indicated by a subscript where needed. The functions $q_k$, $V_k$ and
$h_k$ are real-analytic jointly in the space variables and in $k$: the
right-hand side of \eqref{wformula} is analytic in $(\sigma^2,k)$, and
$\sigma(\nu;k)$ and the expansions at the circle used in the proof of
Proposition~\ref{regularity} depend analytically on $k$ through
\eqref{Rdef}. In particular these functions are bounded in $C^2$ on
compact subsets of $\R^3$, uniformly in $k\in J$, and
$w_*=\sup_{k\in J}\sup_{\nu\ge0}w_k(\nu)$ is finite. It is positive by
\eqref{balance}; numerically $w_*=1.604$, attained at $k=k_c$. The modes
\eqref{explicit_modes} with $n\ge1$ are defined and continuous in $k$ on
all of $J$: there $s_n^2=S+kn^2\le1+kn^2<(n+1)^2$, so that $\alpha_n>0$.

\begin{lemma}\label{modebounds}
Let $k\in J$, with $k<k_c$ if $n=0$.
\begin{enumerate}[label=\textup{(\alph*)},itemsep=1pt,topsep=3pt]
\item $f_n>0$ on $(0,\infty)$ for every $n\ge0$.
\item If $n^2>w_*$ and $\mu_n=\sqrt{n^2-w_*}$, then
 $f_n(\nu)\le(2\sqrt2/\pi)\,Q_{\mu_n-1/2}(\cosh\nu)$ for all $\nu>0$.
\item For every $\nu_1>0$ there is a constant $C$, depending only on
 $\nu_1$ and $J$, such that $\sqrt d\,\sum_{n\ge1}f_n(\nu)\le C$ for
 $\nu\ge\nu_1$ and all $u$.
\end{enumerate}
\end{lemma}
\begin{proof}
(a) On $J$ one has $S>0$, so $s_n$ is real and $\alpha_n\le\beta_n$;
moreover $\alpha_n>0$, for $n=0$ because $S<1$. Euler's integral
\cite{DLMF} gives, for $0\le x<1$,
\begin{equation}\label{Euler}
 \frac{\Gamma(\alpha_n)\Gamma(\beta_n)}{\Gamma(n+1)}\,
 {}_2F_1(\alpha_n,\beta_n;n+1;x)
 =\int_0^1t^{\beta_n-1}(1-t)^{\alpha_n-1}(1-xt)^{-\alpha_n}\,\dd t>0 .
\end{equation}

(b) Put $g=(2\sqrt2/\pi)Q_{\mu_n-1/2}(\cosh\nu)$. It is positive and
satisfies \eqref{modeeq} with $w$ replaced by the constant $w_*$, that
is, $g''+\coth\nu\,g'=(\mu_n^2-\tfrac14)g$.
For $W=\sinh\nu\,(f_n'g-f_ng')$ this and \eqref{modeeq} give
\begin{equation}
 W'=\sinh\nu\,\bigl(w_*-w(\nu)\bigr)f_ng\ge0 .
\end{equation}
Both functions decay exponentially at the circle, so $W(\infty)=0$ and
therefore $W\le0$. Thus $f_n/g$ is non-increasing. Here $n\ge1$, so
$\epsn=2$; by \eqref{lognorm} and the corresponding expansion of the
Legendre function, both $f_n$ and $g$ equal
$-(2\sqrt2/\pi)\log\nu+O(1)$ as $\nu\to0$, so $f_n/g\to1$ there and
$f_n\le g$.

(c) Heine's integral
$Q_{\mu-1/2}(\cosh\nu)=\int_0^\infty
(\cosh\nu+\sinh\nu\cosh t)^{-\mu-1/2}\dd t$ \cite{DLMF}, together with
$\cosh\nu+\sinh\nu\cosh t\ge e^{\nu}$, shows that
$Q_{\mu-1/2}(\cosh\nu)\le e^{-(\mu-\mu')\nu}Q_{\mu'-1/2}(\cosh\nu)$
for $\mu\ge\mu'>0$. Let $n_0$ be the smallest integer with $n_0^2>w_*$.
Since $\mu_n\le n$, one has
$\mu_n-\mu_{n_0}=(n^2-n_0^2)/(\mu_n+\mu_{n_0})\ge n-n_0$, and part (b)
gives
\begin{equation}
 \sum_{n\ge n_0}f_n(\nu)\le\frac{2\sqrt2}{\pi}\,
      \frac{Q_{\mu_{n_0}-1/2}(\cosh\nu)}{1-e^{-\nu}} ,
\end{equation}
and $\sqrt d\,Q_{\mu-1/2}(\cosh\nu)=O(e^{-\mu\nu})$ is bounded on
$\nu\ge\nu_1$. For the finitely many modes with $1\le n<n_0$, note that
$R\le1$ in \eqref{Rdef} gives $\sigma\ge\tanh\nu$, that the left-hand
side of \eqref{Euler} increases with $x$, and that
$d\,\sigma R/\sinh\nu\le\coth(\nu/2)$. On $\nu\ge\nu_1$ therefore
\begin{equation}
 \sqrt d\,f_n\le\frac{\sqrt2}{\pi}\sqrt{\coth\frac{\nu_1}2}\;
 \frac{\Gamma(\alpha_n)\Gamma(\beta_n)}{\Gamma(n+1)}\,
 {}_2F_1\bigl(\alpha_n,\beta_n;n+1;\sech^2\nu_1\bigr),
\end{equation}
and the right-hand side is continuous in $k\in J$.
\end{proof}

\begin{proposition}\label{profile}
Put $\alpha_0=(1-\sqrt S)/2$ and $A_k=[\sqrt2\sin(\pi\alpha_0)]^{-1}$, so
that $A_k(1-S)\to2\sqrt2/\pi$ as $k\uparrow k_c$. For every $\nu_1>0$,
\begin{equation}\label{boundedremainder}
 \sup_{k_1\le k<k_c}\ \sup_{\nu\ge\nu_1}\,
     \bigl|\Phi_k-A_kh_k\bigr|<\infty ,
\end{equation}
and
\begin{equation}\label{profilelimit}
 \boxed{\frac{\Phi_k}{A_k}\longrightarrow h_{k_c}
 \quad\text{in }C^2_{\rm loc}(\R^3)\qquad(k\uparrow k_c).}
\end{equation}
\end{proposition}
\begin{proof}
For $n=0$ one has $\beta_0=1-\alpha_0$ and
$\Gamma(\alpha_0)\Gamma(1-\alpha_0)=\pi/\sin(\pi\alpha_0)$, so
\eqref{explicit_modes} and \eqref{auxh} give exactly
\begin{equation}\label{zeromode}
 \sqrt d\,f_0=A_kh_k\,y_k,\qquad
 y_k={}_2F_1(\alpha_0,1-\alpha_0;1;1-\sigma^2).
\end{equation}
The coefficients of this hypergeometric series satisfy
$(\alpha_0)_j(1-\alpha_0)_j/(j!)^2\le\alpha_0$ for $j\ge1$, whence
$|y_k-1|\le\alpha_0(1-\sigma^2)/\sigma^2$. On $\nu\ge\nu_1$ one has
$\sigma\ge\tanh\nu_1$, independently of $k\in J$. Moreover $A_k\alpha_0\le1/(2\sqrt2)$, and $h_k$ is bounded
uniformly in $k$. Hence $|\sqrt d\,f_0-A_kh_k|$ is bounded on
$\nu\ge\nu_1$, uniformly in $k$. By Theorem~\ref{reconstruction} and
Lemma~\ref{modebounds}(a),
$|\Phi_k-\sqrt d\,f_0|\le\sqrt d\sum_{n\ge1}f_n$, and
Lemma~\ref{modebounds}(c) proves \eqref{boundedremainder}.

Since $A_k\to\infty$ and $h_k\to h_{k_c}$ locally uniformly,
\eqref{boundedremainder} shows that $v_k=\Phi_k/A_k$ converges to
$h_{k_c}$ uniformly on compact sets that do not meet the axis. The
functions $v_k$ are positive solutions of $\Delta_3v_k+V_kv_k=0$ on
$\R^3$, with $V_k$ bounded in $C^1$ on compact sets uniformly in $k$.
Harnack's inequality \cite{GT} bounds $\sup_Bv_k$ by a multiple of
$\inf_Bv_k$ on every ball $B$, with a constant independent of $k$; as
$v_k$ converges at points off the axis, the family is bounded on $B$. The
interior Schauder estimates \cite{GT} then bound it in $C^{2,\alpha}$ on
compact sets. Every sequence $k_j\uparrow k_c$ therefore has a
subsequence converging in $C^2_{\rm loc}(\R^3)$, and its limit coincides
with $h_{k_c}$ off the axis, hence everywhere by continuity. This proves
\eqref{profilelimit}.
\end{proof}

The constant $A_k$ is the value \eqref{Phicircle} of $\Phi_k$ on the
circle, where $h_k=1$. The bounded remainder in \eqref{boundedremainder} is consistent with the
$O(1)$ term in \eqref{masspole}: $8/[\pi(1-S)]-2\sqrt2A_k\to2/\pi$.
For a closed surface the mean curvature
with respect to $g=\Phi^4\hat g$ is
\begin{equation}\label{conformalH}
 H_g=\Phi^{-2}\bigl[H_{\hat g}+4\,\hat n(\log\Phi)\bigr],
\end{equation}
where $\hat g=e^{2q}(\dd\rho^2+\dd z^2)+\rho^2\dd\varphi^2$ is the seed
metric and $\hat n$ its unit outer normal. Only $\nabla\log\Phi$ enters,
so the constant $A_k$ drops out. Define the limiting metric and an
elementary function of $\sigma$, both at $k=k_c$, where $S=1$:
\begin{align}
 g_\infty&=h_{k_c}^4\bigl[e^{2q}(\dd\rho^2+\dd z^2)
            +\rho^2\dd\varphi^2\bigr],\label{ginfty}\\
 \mathcal T(\sigma)&=\frac{\dd}{\dd\nu}
     \log\frac{\sigma^2R^2e^{q}}{\sinh\nu}
 =\frac{3(1-k_c)^2(1-3\sigma^2)-2k_c(5-3k_c)\sigma^4}
        {3(1-k_c)\,\sigma R^3}
  +\frac4{\sqrt{k_c}}\atanh\frac{\sqrt{k_c}\,\sigma}{R}.\label{Tdef}
\end{align}
The second form follows from \eqref{Rdef} and \eqref{qprime}. One has
$\mathcal T=1/(\sqrt{1-k_c}\,\sigma)+O(\sigma)$ at the axis and
$\mathcal T(1)=-1$. Let $\sigma_*$ be the smallest zero of $\mathcal T$
and $\nu_*=\nu(\sigma_*)$; numerically
\begin{equation}\label{sigmastar}
 \sigma_*=0.593813,\qquad\nu_*=0.371001,\qquad
 a\coth(\nu_*/2)=5.4525\,a .
\end{equation}

\begin{theorem}\label{trapped}
\begin{enumerate}[label=\textup{(\roman*)},itemsep=1pt,topsep=3pt]
\item Let $\Sigma\subset\R^3$ be a closed embedded surface of class $C^2$
 whose mean curvature with respect to $g_\infty$ and the outer normal is
 negative. Then there is $k_\Sigma<k_c$ such that $\Sigma$ is outer
 trapped in the data \eqref{metric} for every $k\in(k_\Sigma,k_c)$.
\item This applies to every torus $\nu=\nu_0$ with $0<\nu_0<\nu_*$ and
 to every coordinate sphere $r=r_0$ with $r_0\ge a\coth(\nu_*/2)$.
\item Consequently, for every $k$ sufficiently close to $k_c$ the data
 contain outer trapped tori and outer trapped spheres, and they contain
 an apparent horizon: an outermost minimal surface, each component of
 which is a sphere, enclosing every surface of \textup{(i)} that is
 outer trapped for that $k$.
\end{enumerate}
\end{theorem}
\begin{proof}
(i) By \eqref{conformalH}, $\Phi_k^2H_{g_k}
=H_{\hat g_k}+4\hat n_k(\log(\Phi_k/A_k))$ on $\Sigma$. By
\eqref{profilelimit} and the continuity of $q_k$ in $k$, the right-hand
side converges uniformly on the compact surface $\Sigma$ to
$H_{\hat g_{k_c}}+4\hat n_{k_c}(\log h_{k_c})=h_{k_c}^2H_{g_\infty}<0$.

(ii) The torus $\nu=\nu_0$ bounds the solid torus $\nu>\nu_0$ around the
circle, so its outer normal points towards decreasing $\nu$. By
\eqref{gtor} its area element is
$a^2\Phi^4d^{-2}e^{q}\sinh\nu\,\dd u\,\dd\varphi$, and
$h_{k_c}^4d^{-2}=\sigma^2R^2/\sinh^2\nu$ by \eqref{auxh}. Hence, in
$g_\infty$,
\begin{equation}\label{torusH}
 H=-\frac{d\,e^{-q}}{a\,h_{k_c}^2}\,\mathcal T(\sigma),
\end{equation}
which is negative for $\sigma<\sigma_*$; it does not change sign along
the torus. For the sphere $r=r_0$, \eqref{conformalH} gives
$H=h_{k_c}^{-2}e^{-q}r_0^{-1}\,\Xi$ with
$\Xi=4r\partial_r\log h_{k_c}+r\partial_rq+2$. In toroidal
coordinates the dilation field is
$r\partial_r=-\sinh\nu\cos u\,\partial_\nu-\cosh\nu\sin u\,\partial_u$,
and a short calculation with \eqref{auxh} gives
\begin{equation}\label{spheresign}
 \Xi=-\cos u\,\bigl[\cosh\nu+\sinh\nu\,\mathcal T(\sigma)\bigr].
\end{equation}
Since $r^2/a^2=(\cosh\nu+\cos u)/(\cosh\nu-\cos u)$, one has $\cos u>0$
for $r>a$. On the sphere $r=r_0$ the coordinate $\nu$ is largest on the
equator, where $\tanh\nu=2ar_0/(r_0^2+a^2)$; this does not exceed
$\tanh\nu_*$ when $r_0\ge a\coth(\nu_*/2)$. Then
$\sinh\nu\,\mathcal T\ge0$ on the sphere, with the limiting value $1$ at
its poles, and $\Xi\le-\cos u\cosh\nu<0$.

(iii) Let $\Sigma$ be as in (i) and $k\in(k_\Sigma,k_c)$. Large
coordinate spheres have $H>0$ by \eqref{AF} and \eqref{q_asym}. The
compact region between $\Sigma$ and such a sphere has an outer trapped
inner boundary and an outer untrapped outer boundary, so it contains a
stable minimal surface enclosing $\Sigma$. This follows by minimizing
area between the two barriers, and is the time-symmetric case of the
existence theorem for marginally outer trapped surfaces
\cite{AnderssonMetzger,Eichmair}; no condition on the topology of the
inner boundary is needed, so that $\Sigma$ may be a torus. The data
therefore contain compact minimal surfaces. By
\cite[Lemma~4.1]{HuiskenIlmanen} the trapped region, which contains
every compact minimal surface and the bounded components of its
complement, and hence $\Sigma$, is bounded by an outermost minimal
surface whose components are spheres; this surface minimizes area among
the surfaces enclosing it and is therefore stable.
\end{proof}

\begin{remark}\label{limitgeometry}
By \eqref{profilelimit}, the rescaled metrics $A_k^{-4}g_k$ converge on
compact sets to $g_\infty$, which is scalar-flat. The inversion
$r\mapsto a^2/r$, that is $u\mapsto\pi-u$, is an isometry of $g_\infty$,
because $h_{k_c}$ is mapped to $(r/a)h_{k_c}$ and the seed metric to
$(a/r)^4$ times itself. Hence $g_\infty$ extends smoothly over the point
at infinity to a scalar-flat metric on $S^3$, and the sphere $r=a$ is
totally geodesic, in particular minimal. The wave thus nearly pinches
off a closed space, joined to the asymptotically flat region by a neck
at $r\approx m/2$ (Section~\ref{sec_trapped}) that recedes to infinity
in these coordinates. A numerical evaluation gives
$\cosh\nu+\sinh\nu\,\mathcal T>0$ for all $\nu$ and a single zero of
$\mathcal T$; by \eqref{spheresign} and \eqref{torusH} every sphere
$r_0>a$ is then outer trapped in $g_\infty$, and $\nu=\nu_*$ is its only
minimal torus of this form, the counterpart of the marginally outer
trapped tori of \cite{Husa96}. The theorem uses only the elementary
facts stated in its proof. The parameter $k_\Sigma$ depends on the
surface, and the theorem gives no estimate of the size of the left
neighbourhood of $k_c$ in which trapped surfaces exist.
\end{remark}

\section{Trapped surfaces: a restricted numerical search}\label{sec_trapped}
A minimal surface
found in a prescribed ansatz need not be globally outermost. Here the
search is restricted to equatorially symmetric surfaces of revolution that
are radial graphs $r=r(\theta)$; below, ``minimal sphere'' and
``marginally trapped sphere'' are used interchangeably.

The area functional is
\begin{equation}
 \mathcal A=2\pi\int\Phi^4e^q\rho\,\dd l
 =2\pi\int_0^\pi\Phi^4e^q r^2\sin\theta
                   \sqrt{1+s'^2}\,\dd\theta,\qquad s=\log(r/a).
\end{equation}
Its Euler--Lagrange equation is
\begin{equation}\label{horizonode}
 s''=(1+s'^2)(\mathcal F_s-\mathcal F_\theta s'),\qquad
 \mathcal F=4\log\Phi+q+2s+\log\sin\theta,
\end{equation}
with
\begin{equation}\label{horizonbc}
 s'(0)=s'(\pi/2)=0,\qquad s''(0)=\tfrac12\mathcal F_s|_{\theta=0}.
\end{equation}
The $q$ contribution is retained. For each $s_0=s(0)$, integration gives
the mismatch $\mu(s_0;k)=s'(\pi/2)$. A fold satisfies
$\mu=\partial_{s_0}\mu=0$.

\subsection{Evaluation near the axis}
The regularized representation \eqref{correctedseries} is used. The leading
local behaviour is
\begin{equation}\label{cusp}
 D=D(0,0)-\frac{w(0)}{\sqrt2}\sqrt{\nu^2+u^2}
                                  +O(\nu^2+u^2).
\end{equation}
It follows from the next asymptotic term \eqref{nextorder}, proved in
Appendix~\ref{app_asym}, or directly from \eqref{F_eq}. Subtracting
$M=-w(0)\sqrt d$ removes this conical term. If $m_n$ are its cosine
coefficients, then $e_n=d_n-m_n$ have $e_n(0)=O(n^{-4})$.
At $\nu=0$,
\begin{equation}
 m_0=-\frac{2\sqrt2w(0)}\pi,\qquad
 m_n=\frac{4\sqrt2w(0)}{\pi(4n^2-1)}\quad(n\ge1).
\end{equation}
For $\nu>0$ the same coefficients follow by multiplying the Heine series
by $\cosh\nu-\cos u$; in the accompanying code they are computed by a
fast Fourier transform of $M$.

The zero-balanced connection series behind \eqref{lognorm} is numerically
reliable only for $n\sigma\lesssim5$: its terms grow like $e^{n\sigma}$,
whereas the mode decays like $e^{-n\sigma}$. The modes are therefore
tabulated by integrating \eqref{modeeq} from the circle towards the axis,
which is the stable direction, on a logarithmic grid
$10^{-4}\le\nu\le14$, and are normalized with the connection series at a
point where $n\sigma\approx\tfrac12$. The tables agree with
\eqref{explicit_modes}, evaluated in 30-digit arithmetic, to a relative
accuracy of $10^{-9}$ where the modes are of order one and of $10^{-6}$
in their exponentially small tails. The Legendre functions in
\eqref{correctedseries} are obtained by the same procedure with $k=0$,
and $q$ by integrating $w$ twice. At test points with $\nu$ between
$0.7$ and $2$ the resulting $\Phi$ agrees with the direct summation of
\eqref{Phiseries} to $10^{-9}$.

\subsection{Numerical findings and their scope}
The shooting implementation uses $N=1200$ modes after cusp subtraction.
It starts at $\theta_1=0.01$ with the Taylor values
$s=s_0+\tfrac12s''(0)\theta_1^2$, $s'=s''(0)\theta_1$, and uses an adaptive
Runge--Kutta integrator with relative tolerance $10^{-9}$; derivatives
of $\mathcal F$ are central differences with step $10^{-5}$. Applied to
the Schwarzschild conformal factor $1+m/(2r)$, the same code returns
$r=m/2$ with a relative error of $2\times10^{-11}$.

For every $k$ examined, $\mu(s_0;k)$ is positive for small and for large
$s_0$ and has a single minimum, which decreases with $k$. A fold of
minimal spheres is found at
\begin{equation}\label{fold}
 k_h\simeq0.767019,\qquad m(k_h)/a\simeq13.1286.
\end{equation}
At the fold $\mu=\partial_{s_0}\mu=0$. It is an ordinary turning point:
central differences give
\begin{equation}\label{nondegenerate}
 \partial_k\mu=-21.81,\qquad \partial_{s_0}^2\mu=0.834
\end{equation}
there, with steps $10^{-4}$ to $4\times10^{-4}$ in $k$ and $0.01$ to
$0.04$ in $s_0$ giving the same digits. Both are non-zero, so the fold is
non-degenerate and two solutions exist for $k$ slightly above $k_h$ and
none slightly below.

Table~\ref{tableconvergence} separates the two sources of numerical
error. The truncation of the mode sum is converged to the digits shown
for $N\ge400$; with $N=60$ or $80$ the sixth significant digit of the
radii and of the area is not yet settled. The integrator
tolerance has no visible effect at fixed $N$. The values in \eqref{fold}
are also unchanged for $\theta_1=0.005$ and $0.02$. The position $r(0)$
of the marginal surface is determined less accurately than $k_h$, because
$\mu$ is stationary in $s_0$ there; it is quoted as $3.424$, with an
uncertainty of about $2\times10^{-4}$. These are convergence tests of a
numerical calculation. They are not interval certificates or proofs of
global outermostness.

\begin{table}[tb]
\centering
\begin{tabular}{rrrrrrr}
\toprule
 &\multicolumn{2}{c}{fold}&\multicolumn{2}{c}{$k=0.78$}&\multicolumn{2}{c}{$k=0.80$}\\
$N$&$k_h$&$r(0)/a$&$r(\pi/2)/a$&$\mathcal A/a^2$&$r(\pi/2)/a$&$\mathcal A/a^2$\\\midrule
$40$&$0.7670210$&$3.4232$&$9.355980$&$23361.147$&$63.22536$&$810771.60$\\
$60$&$0.7670187$&$3.4249$&$9.355059$&$23360.981$&$63.22450$&$810804.23$\\
$80$&$0.7670185$&$3.4243$&$9.354951$&$23361.017$&$63.22372$&$810808.65$\\
$200$&$0.7670185$&$3.4240$&$9.355008$&$23361.015$&$63.22323$&$810808.20$\\
$400$&$0.7670185$&$3.4240$&$9.355009$&$23361.014$&$63.22323$&$810808.23$\\
$800$&$0.7670185$&$3.4239$&$9.355009$&$23361.014$&$63.22323$&$810808.23$\\
$1200$&$0.7670185$&$3.4239$&$9.355009$&$23361.014$&$63.22323$&$810808.23$\\
$1600$&$0.7670185$&$3.4240$&$9.355009$&$23361.014$&$63.22323$&$810808.23$\\\midrule
\multicolumn{7}{l}{$N=1200$, integrator tolerance $10^{-7}$ (difference step $10^{-4}$):}\\
 &$0.7670185$&$3.4240$&$9.355009$&$23361.014$&$63.22323$&$810808.23$\\
\multicolumn{7}{l}{$N=1200$, integrator tolerance $10^{-11}$ (difference step $10^{-6}$):}\\
 &$0.7670185$&&$9.355009$&$23361.014$&&\\
\bottomrule
\end{tabular}
\caption{Convergence of the fold parameter, of the axis radius of the
marginal surface, and of the equatorial radius and area of the
outer-branch surfaces at $k=0.78$ and $k=0.80$. Upper part: number of
modes $N$ at integrator tolerance $10^{-9}$. Lower part: integrator
tolerance at fixed $N$; entries not computed are left blank.}\label{tableconvergence}
\end{table}

Two radial-graph solutions occur on the studied side above this fold.
The larger solution is an outer-branch marginally trapped sphere in this
ansatz. It bounds a trapped region: for $k=0.77$, $0.78$ and $0.80$ the
surfaces obtained from it by the radial displacements
$\delta\log r=-0.01$ and $-0.05$ have $H<0$ at every point, so they are
outer trapped, while the displacements $+0.01$ and $+0.05$ give $H>0$
everywhere.
Table~\ref{tablehorizons} records representative geometrical quantities,
and Fig.~\ref{fighorizons} shows three of the surfaces. They are slightly
prolate in the coordinates $(\rho,z)$ and lie far from the circle
$\rho=a$: the largest value of $\nu$ on the marginal surface is $0.68$.

\begin{table}[tb]
\centering
\begin{tabular}{rrrrrr}
\toprule
$k$&$m/a$&$r(0)/a$&$r(\pi/2)/a$&$\mathcal A/a^2$
 &$\sqrt{\mathcal A/(16\pi)}/m$\\\midrule
$0.767019$&$13.1286$&$3.424$&$3.058$&$8.174\times10^3$&$0.97134$\\
$0.77$&$14.5151$&$5.271$&$4.802$&$1.0188\times10^4$&$0.98081$\\
$0.78$&$21.6636$&$9.864$&$9.355$&$2.3361\times10^4$&$0.99513$\\
$0.79$&$38.6505$&$18.919$&$18.490$&$7.4994\times10^4$&$0.99936$\\
$0.80$&$127.007$&$63.425$&$63.223$&$8.1081\times10^5$&$0.999993$\\
\bottomrule
\end{tabular}
\caption{Numerical outer-branch marginally trapped (minimal) spheres in the radial-graph search:
ADM mass, coordinate radius on the axis and in the equatorial plane, area,
and the ratio of $\sqrt{\mathcal A/(16\pi)}$ to the ADM mass. The first row
is the marginal surface at the fold \eqref{fold}.}\label{tablehorizons}
\end{table}

\begin{figure}[tb]
\centering\includegraphics[width=0.6\linewidth]{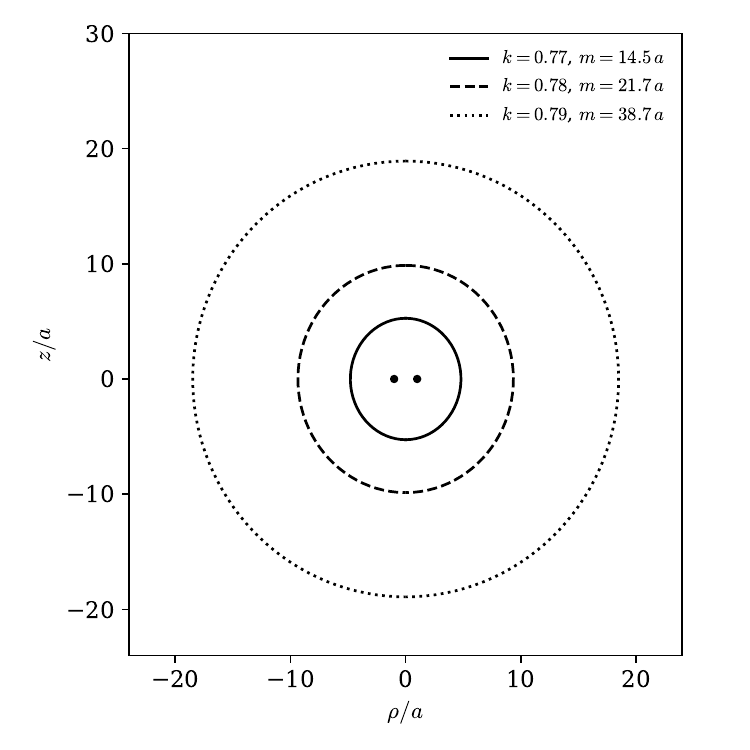}
\caption{Outer-branch marginally trapped spheres of the radial-graph search in the
$(\rho,z)$ plane for three values of $k$. The dots mark the circle
$\rho=a$, $z=0$.}\label{fighorizons}
\end{figure}

The numerical area ratios are compatible with the Riemannian Penrose
inequality \cite{HuiskenIlmanen,Bray}, tested earlier on numerical Brill
waves in \cite{KKS}; the deficit is largest at the fold,
where it is $2.9\%$. The theorem concerns the appropriate
outermost horizon; agreement of these ratios is a consistency check and does
not itself establish outermostness. As $k\uparrow k_c$, the calculated outer
branch approaches $r=m/2$ and becomes nearly spherical in these coordinates.

The surfaces of Theorem~\ref{trapped} can be followed in the exact data
well below $k_c$. Coordinate spheres $r=r_0$ are outer trapped for
$k\gtrsim0.777$, first at $r_0\approx4.2\,a$, and tori $\nu=\nu_0$ for
$k\gtrsim0.800$, first at $\nu_0\approx0.17$; at $k=0.804$ the tori with
$0.05\le\nu_0\le0.3$ are outer trapped, to be compared with
$\nu_*=0.371$ in \eqref{sigmastar}. The inner radial-graph solution
approaches the sphere $r=a$, which is minimal in the limiting metric
\eqref{ginfty} (Remark~\ref{limitgeometry}): its axis and equatorial
radii are $(1.630,1.498)\,a$ at $k=0.78$, $(1.082,1.063)\,a$ at $k=0.80$
and $(1.013,1.010)\,a$ at $k=0.804$.

For comparison, for conformally flat data whose conformal factor is one
plus the Newtonian potential of a thin ring of radius $a$, the same code
gives the corresponding fold at $m\simeq2.03\,a$. The field of the
present, spatially extended wave is much less compact.

The threshold \eqref{fold} is not asserted to be the first occurrence of
any marginally trapped surface in the complete manifold. Non-star-shaped
surfaces and other topologies, such as the toroidal trapped surfaces
studied in \cite{Husa96,KarkowskiTori}, have not been exhaustively
searched. In particular, no absence theorem for horizons is claimed for
$0<k<k_h$; the branch $k<0$ is treated in Section~\ref{sec_negative}.

\subsection{The branch $k<0$}\label{sec_negative}
For $k<0$ one has $S(k)<1$ (Lemma~\ref{critical}), so this branch has no
critical point, and its mass grows only slowly with $|k|$
(Table~\ref{tablemass}): $m/a=1.2452$ at $k=-100$ and $1.8712$ at
$k=-1000$. On part of the branch trapped surfaces can be excluded. For
the coordinate sphere $r=r_0$, \eqref{conformalH} gives
\begin{equation}\label{sphereH}
 H_g=\Phi^{-2}e^{-q}r_0^{-1}\,\bigl[2+r\partial_r(q+4\log\Phi)\bigr],
\end{equation}
where $r\partial_r$ is the dilation field used in the proof of
Theorem~\ref{trapped}; the sphere is outer trapped if and only if the
bracket is negative on the whole of it. By the area element given in
that proof, the torus $\nu=\nu_0$ is outer trapped if and only if
$\partial_\nu\log\bigl(\Phi^4d^{-2}e^{q}\sinh\nu\bigr)>0$ for all $u$.

A numerical evaluation of the bracket in \eqref{sphereH} on a grid in
$(\nu,u)$ that covers all spheres shows that it is positive for
$-8\le k<0$. Its minimum, attained in the equatorial plane just outside
the circle, decreases from $2$ at $k=0$ to $1.80$ at $k=-3$, $1.17$ at
$k=-6$ and $0.30$ at $k=-8$, and vanishes at $k\approx-8.5$. For
$-8\le k<0$ every coordinate sphere is therefore mean convex, and a
closed surface with $H\le0$ would touch its smallest enclosing
coordinate sphere from inside, in contradiction with the maximum
principle: these data contain no trapped surface, no minimal surface
and no apparent horizon. For $k<-8.5$ the bracket is negative in a small
region just outside the circle, and the argument does not apply. For the
values examined there, $-30\le k\le-9$, no coordinate sphere is outer
trapped, since on each of them the bracket exceeds $1.6$ near the axis,
and no torus $\nu=\nu_0$ is outer trapped, since the derivative above is
negative somewhere on each torus. More negative $k$ and other surfaces
have not been examined.

\section{A scaling symmetry, further solvable families and the simplest member}\label{sec_scaling}
The family of Section~\ref{sec_family} is the first member of a sequence.
The mode equation \eqref{schrodinger} has a scaling symmetry which turns
every admissible potential into further admissible potentials, and the
simplest member of the resulting sequence has an elementary seed function
in cylindrical coordinates. Throughout this section $\ell$ denotes a
positive integer.

\subsection{Rescaled families}
Let $\tilde w_{k,S}$ denote the potential \eqref{wformula} with the
constant $S$ left free, and let
$W_{k,S}(\sigma)=\tilde w_{k,S}+1/(4\sinh^2\nu)$ be its regular part, the
first two terms of \eqref{wformula}. Define
\begin{equation}\label{wscaled}
 w_{\ell,k}(\nu)=\ell^2\Bigl[\tilde w_{k,S}(\ell\nu)+\frac1{4\sinh^2(\ell\nu)}\Bigr]
 -\frac1{4\sinh^2\nu}
 =\ell^2\,W_{k,S}\bigl(\sigma_k(\ell\nu)\bigr)-\frac1{4\sinh^2\nu},
\end{equation}
where $\sigma_k$ is the variable of \eqref{Rdef}--\eqref{nusigma}; for
$\ell=1$ this is \eqref{wformula}.

\begin{proposition}[scaling]\label{scaling}
Let $\ell\ge1$ be an integer and $k<1$.
\begin{enumerate}[label=(\roman*),leftmargin=2em]
\item $w_{\ell,k}$ is real analytic and even at the axis, and analytic in
$e^{-2\nu}$ with $w_{\ell,k}=O(e^{-2\nu})$ at the circle, so that
$V=(d^2/a^2)w_{\ell,k}$ is real analytic on $\R^3$.
\item In the variable $\tilde\nu=\ell\nu$ the mode equation
\eqref{schrodinger} with the potential $w_{\ell,k}$ and the index $n$ is
the mode equation with the potential $\tilde w_{k,S}$ and the index
$\tilde n=n/\ell$. Every mode equation is therefore hypergeometric, with
the parameters \eqref{ab} evaluated at $\tilde n$,
\begin{equation}\label{abscaled}
 \tilde n=\frac n\ell,\qquad
 \tilde s_n=\sqrt{S+k\tilde n^2},\qquad
 \tilde\alpha_n,\tilde\beta_n=\frac{1+\tilde n\mp\tilde s_n}2 ,
\end{equation}
and the circle-regular mode with the normalization of Section~\ref{sec_modes}
is
\begin{equation}\label{modescaled}
 f_n=\frac{\epsn}{\sqrt2\pi}
 \frac{\Gamma(\tilde\alpha_n)\Gamma(\tilde\beta_n)}{\Gamma(\tilde n+1)}
 \sqrt{\frac{\sigma R}{\ell\sinh\nu}}\,(1-\sigma^2)^{\tilde n/2}\,
 {}_2F_1\bigl(\tilde\alpha_n,\tilde\beta_n;\tilde n+1;1-\sigma^2\bigr),
 \qquad \sigma=\sigma_k(\ell\nu).
\end{equation}
\item The circle condition \eqref{balance} holds if and only if
\begin{equation}\label{Sscaled}
 S=S_\ell(k)=\frac{\sqrt k}{\atanh\sqrt k}
 \left[\frac{3+k}{12(1-k)}+\frac{\ell-1}{4\ell}\right],
\end{equation}
with the continuation $\atanh\sqrt k/\sqrt k=\arctan\sqrt{|k|}/\sqrt{|k|}$
for $k<0$; $S_1=S$ is \eqref{Sdef}.
\item With $S=S_\ell(k)$, the seed function is
\begin{equation}\label{qscaled}
 q_{\ell,k}(\nu)=\tilde q_{k,S}(\ell\nu)+\log\frac{\ell\sinh\nu}{\sinh(\ell\nu)},
 \qquad
 q_{\ell,k}'(\nu)=\coth\nu-\frac{\ell P(\sigma)}{3\sigma(1-k)R^3}+4\ell S A(\sigma),
\end{equation}
where $\tilde q_{k,S}$ is the closed form \eqref{qclosed} with $S$ in
place of $S(k)$ (the formula is linear in $S$), and $P$, $A$ are
\eqref{P}--\eqref{Adef}. It satisfies the boundary conditions \eqref{bcq},
and $q$, $V$ are real analytic on $\R^3$.
\item The function $h=\sqrt d\,\sqrt{\sigma R/(\ell\sinh\nu)}$,
$\sigma=\sigma_k(\ell\nu)$, satisfies
\begin{equation}\label{scalingid}
 (\Delta_3+V)h=(S-1)\,\ell^2\,\frac{d^2(1-\sigma^2)}{a^2R^2}\,h ,
\end{equation}
the analogue of \eqref{auxeq}.
\end{enumerate}
\end{proposition}

\begin{proof}
(ii) With $\tilde\nu=\ell\nu$ and $g(\nu)=\tilde g(\tilde\nu)$,
\eqref{schrodinger} becomes
$\tilde g''+\ell^{-2}[-n^2+1/(4\sinh^2(\tilde\nu/\ell))+w_{\ell,k}(\tilde\nu/\ell)]\tilde g=0$,
and by \eqref{wscaled} the bracket equals
$\ell^2[-\tilde n^2+1/(4\sinh^2\tilde\nu)+\tilde w_{k,S}(\tilde\nu)]$. The
reduction of Section~\ref{sec_modes} to \eqref{hypergeom} uses the index
only through $\tilde n^2$ in the constant term and through the exponent
$\mu=\tilde n$ at the circle, so it holds for real $\tilde n$. Undoing
the transformations and dividing by $\sqrt\ell$, $f_n(\nu)=\sqrt{\sinh(\ell\nu)/(\ell\sinh\nu)}\,
\tilde f_{\tilde n}(\ell\nu)$ is a circle-regular solution, with $\tilde f$ the function
\eqref{explicit_modes} at the index $\tilde n$. At the axis the prefactor
tends to $1$ and $\log(\ell\nu)=\log\nu+\log\ell$, so the coefficient of
$-\log\nu$ in \eqref{lognorm} is unchanged and the constant of
\eqref{explicit_modes} carries over, which gives \eqref{modescaled}.

(i) By Proposition~\ref{regularity}, $W_{k,S}(\sigma_k(\tilde\nu))$ is
$1/(4\tilde\nu^2)$ plus an even analytic function of $\tilde\nu$; hence
$\ell^2W_{k,S}(\sigma_k(\ell\nu))-1/(4\sinh^2\nu)$ is even analytic at
$\nu=0$, with $w_{\ell,k}(0)=\ell^2\tilde w_{k,S}(0)+(1-\ell^2)/12$. At
the circle, $1-\sigma_k(\tilde\nu)^2$ is an analytic function of
$e^{-2\tilde\nu}$ with a simple zero (proof of Proposition~\ref{regularity}), so
$W_{k,S}(\sigma_k(\ell\nu))$ is analytic in $e^{-2\ell\nu}$ and vanishes
there; an analytic function of $e^{-2\ell\nu}$ is analytic in
$e^{-2\nu}$ because $\ell$ is an integer, and
$\ell^2/(4\sinh^2(\ell\nu))-1/(4\sinh^2\nu)$ is analytic in $e^{-2\nu}$
with a simple zero. For non-integer $\ell$ this fails, and $V$ would not
be smooth at the circle. There the modes behave as
$f_n=e^{-(n+1/2)\nu}\times(\text{analytic in }e^{-2\nu})$: the factor
$\tilde f_{\tilde n}(\ell\nu)$ contributes $e^{-(\tilde n+1/2)\ell\nu}$
and the prefactor $e^{(\ell-1)\nu/2}$, and $-(\tilde n+\half)\ell+\half(\ell-1)
=-(n+\half)$. The regularity of $q$, $V$ and $\Phi$ then follows as for
$\ell=1$.

(iii) Both terms of \eqref{wscaled} carry the non-integrable
$1/(4\nu^2)$ at the axis; combining them,
\begin{equation}
 \int_0^\infty w_{\ell,k}\,\dd\nu
 =\ell\int_0^\infty\tilde w_{k,S}\,\dd\tilde\nu
 +\Bigl[\tfrac14\coth\nu-\tfrac\ell4\coth(\ell\nu)\Bigr]_0^\infty
 =\ell\int_0^\infty\tilde w_{k,S}\,\dd\tilde\nu+\frac{1-\ell}4 .
\end{equation}
By \eqref{qprime} at $\sigma\to1$,
$4\int_0^\infty\tilde w_{k,S}=-\frac{3+k}{3(1-k)}+4S\,\atanh\sqrt k/\sqrt k$,
and the vanishing of the sum is \eqref{Sscaled}. For $\ell\ge2$ the
rescaled potential does not satisfy the circle condition in
$\tilde\nu$: $\tilde q'(\infty)=(\ell-1)/\ell$ and $\tilde q(\ell\nu)$
grows like $(\ell-1)\nu$, which is exactly cancelled by
$\log(\ell\sinh\nu/\sinh(\ell\nu))\simeq\log\ell-(\ell-1)\nu$ in
\eqref{qscaled}.

(iv) Integrating $4w_{\ell,k}$ with the same combination of the singular
terms gives the second formula of \eqref{qscaled}; a second integration
gives $\tilde q_{k,S}(\ell\nu)-\log\sinh(\ell\nu)+\log\sinh\nu+\text{const}$,
and $q(0)=0$ fixes the constant to $\log\ell$.

(v) In the variable $\tilde\nu$ the function $\sqrt{\sigma R}$ satisfies
the one-dimensional form of \eqref{auxeq}; multiplying by $\ell^2$ and
using \eqref{wscaled} gives the same identity in $\nu$ with the
right-hand side multiplied by $\ell^2$, and \eqref{sqrtd} turns it into
\eqref{scalingid}.
\end{proof}

The identity \eqref{scalingid} has the form of \eqref{auxeq} with a
positive factor, and $h\sim\sqrt{\ell d}$ at infinity. The proofs of
Lemma~\ref{auxlemma}, Lemma~\ref{critical} and Theorem~\ref{existence}
use only these properties, so they apply with $h_k$ replaced by $h$:
complete asymptotically flat data with $\Phi>0$ exist precisely for
$S_\ell(k)<1$, and the reconstruction of Section~\ref{sec_modes}
identifies them with $\Phi=\sqrt d\sum_nf_n\cos(nu)$, the modes being
\eqref{modescaled}. The mass is \eqref{mass} with $\alpha_n,\beta_n$
replaced by $\tilde\alpha_n,\tilde\beta_n$ and $\log\frac{1-k}4$ by
$\log\frac{1-k}{4\ell^2}$, and on the circle
$\Phi=1/\bigl(\sqrt{2\ell}\cos(\pi\sqrt S/2)\bigr)$ in place of
\eqref{Phicircle}.

For $\ell=1$ the condition $S_1(k)<1$ is $k<k_c$ (Lemma~\ref{critical}).
For $\ell\ge2$ it is two-sided. As $k\to-\infty$,
$\sqrt k/\atanh\sqrt k=\sqrt{|k|}/\arctan\sqrt{|k|}\simeq2\sqrt{|k|}/\pi$,
while the bracket in \eqref{Sscaled} tends to $(2\ell-3)/(12\ell)$,
negative for $\ell=1$ but positive for $\ell\ge2$. Hence for $\ell\ge2$
$S_\ell(k)\to+\infty$ at both ends of $(-\infty,1)$, and $S_\ell(k)<1$
holds exactly on the interval $k_c^{(\ell,-)}<k<k_c^{(\ell)}$ between
the two roots of $S_\ell(k)=1$; numerically $S_\ell$ has a single
minimum, at negative $k$ (for $\ell=2$ it is $0.180$, at $k=-9.84$). The
critical parameters are
\begin{center}
\small\setlength{\tabcolsep}{4.5pt}
\begin{tabular}{r|cccccccc|c}
$\ell$ & 1 & 2 & 3 & 4 & 5 & 10 & 20 & 50 & $\infty$\\\hline
$k_c^{(\ell)}$ & 0.80480 & 0.78337 & 0.77493 & 0.77041 & 0.76760 & 0.76172 & 0.75865 & 0.75676 & 0.75548\\
$k_c^{(\ell,-)}$ & --- & $-1356.49$ & $-322.56$ & $-201.05$ & $-157.67$ & $-103.36$ & $-85.76$ & $-77.15$ & $-72.06$
\end{tabular}
\end{center}
The positive ones decrease with $\ell$. These values are coordinates,
not amplitudes. A better measure is the minimum of the seed function at
$k_c^{(\ell)}$, which is $-1.92$, $-1.35$,
$-1.16$, $-1.04$ and $-0.95$ for $\ell=1,2,3,5,10$: the wider waves
(large $\ell$ spreads $q$ far beyond the circle) form trapped surfaces at
a smaller amplitude, an exact instance of the dependence of the critical
amplitude on the shape of the seed function observed numerically
\cite{Eppley,Holz}.

\begin{remark}
The set $(\ell,k)$ is what Natanzon's hypergeometric class allows. In
\eqref{natanzon}, $\mathcal R(0)=0$ is forced by the axis, and the
analyticity of $V$ in $e^{-2\nu}$ at the circle requires
$1-\zeta\sim e^{-2\ell\nu}$ with an integer $\ell$, that is,
$\mathcal R(1)=1/\ell^2$; the substitution $\tilde\nu=\ell\nu$ then maps
$\mathcal R=\zeta(a_1+a_2\zeta)$, $a_1+a_2=1/\ell^2$, to
$\zeta(1-\tilde k+\tilde k\zeta)$ with $\tilde k=\ell^2a_2$. Section~\ref{sec_family}
is the case $\ell=1$.
\end{remark}

\subsection{The simplest explicit Brill wave}
At $k=0$ one has $\sigma=\tanh\tilde\nu$, $R=1$,
$\tilde w_{0,S}=(S-\tfrac14)\sech^2\tilde\nu$ and
$\tilde q_{0,S}=(4S-1)\log\cosh\tilde\nu$, the $k\to0$ limit of
\eqref{qclosed} (directly: $\tilde q'=\coth\tilde\nu-(1+\sigma^2)/\sigma+4S\sigma
=(4S-1)\tanh\tilde\nu$). Then \eqref{Sscaled} gives $S=(2\ell-1)/(4\ell)$,
and with $T=\tanh\nu=2a\rho/(r^2+a^2)$,
\begin{align}
 w_\ell(\nu)&=\frac{\ell^2}{4\sinh^2(\ell\nu)}-\frac1{4\sinh^2\nu}
 +\frac{\ell(\ell-1)}4\,\sech^2(\ell\nu),\label{welem}\\
 q_\ell(\nu)&=\frac{\ell-1}\ell\log\cosh(\ell\nu)+\log\frac{\ell\sinh\nu}{\sinh(\ell\nu)}
 =\frac{\ell-1}\ell\log P_\ell(T^2)-\log\frac{Q_\ell(T^2)}\ell ,\label{qelem}
\end{align}
where $\cosh(\ell\nu)=\cosh^\ell\!\nu\,P_\ell(T^2)$ and
$\sinh(\ell\nu)=\sinh\nu\cosh^{\ell-1}\!\nu\,Q_\ell(T^2)$ define the
polynomials $P_\ell=\sum_j\binom\ell{2j}T^{2j}$ and
$Q_\ell=\sum_j\binom\ell{2j+1}T^{2j}$. Every member of this sequence is
elementary and, through $T^2=4a^2\rho^2/(\rho^2+z^2+a^2)^2$,
rational-logarithmic in the cylindrical coordinates; since $S<\half$,
all are subcritical. The first non-flat member, $\ell=2$, is the simplest
explicit Brill wave known to the author:
\begin{equation}\label{simplest}
 \boxed{\;q=\tfrac12\log\bigl(1+\tanh^2\nu\bigr)
 =\tfrac12\log\Bigl[1+\frac{4a^2\rho^2}{(\rho^2+z^2+a^2)^2}\Bigr],\qquad
 V=\frac{a^2\bigl(s^4-4ps^2-p^2\bigr)}{s^2\,(s^2+p)^2},\;}
\end{equation}
with $s=\rho^2+z^2+a^2$ and $p=4a^2\rho^2$; equivalently
$w=\half\sech^2(2\nu)-\tfrac14\sech^2\nu$ and $S=\tfrac38$. One checks
directly that $q=\partial_\rho q=0$ on the axis, $q\simeq2a^2\rho^2/r^4$
at infinity, $q=\half\log2$ on the circle, and
$\int_0^\infty w\,\dd\nu=\half\cdot\half-\tfrac14\cdot1=0$. The modes
are \eqref{modescaled} with $\sigma=\tanh2\nu$:
\begin{equation}\label{simplestmodes}
 f_n=\frac{\epsn}{\sqrt2\pi}
 \frac{\Gamma(\tilde\alpha_n)\Gamma(\tilde\beta_n)}{\Gamma(\tfrac n2+1)}
 \sqrt{\frac{\tanh2\nu}{2\sinh\nu}}\,(\sech2\nu)^{n/2}\,
 {}_2F_1\Bigl(\tilde\alpha_n,\tilde\beta_n;\tfrac n2+1;\sech^22\nu\Bigr),
\end{equation}
with $\tilde\alpha_n,\tilde\beta_n=(4+2n\mp\sqrt6)/8$, and
$\Phi=\sqrt d\sum_nf_n\cos(nu)$. The value on the circle and the mass
are
\begin{equation}\label{simplestnumbers}
 \Phi\big|_{\rho=a,\,z=0}=\frac1{2\cos\bigl(\pi\sqrt{3/8}/2\bigr)}=0.87420,
 \qquad m=0.039390\,a ,
\end{equation}
and $\Phi<1$ in the whole region $r\lesssim3a$, with minimum $0.8595$
near the circle; Figure~\ref{figsimplest} shows the fields. The data are weak ($m/a$ is small, $\max q=\half\log2$) but
finite-amplitude and exact, with a seed function that any code can
evaluate; they may be the most convenient exact benchmark in this
article. They contain no trapped surfaces. By \eqref{conformalH}, the
mean curvature of the coordinate sphere $r=r_0$ is
$H=\Phi^{-2}e^{-q}r_0^{-1}\,[2+r\partial_r(q+4\log\Phi)]$, and a
numerical evaluation of the bracket on all spheres gives the minimum
$1.996$, approached as $2-2m/r$ at large $r$, while near the wave the
bracket exceeds $2$. Every coordinate sphere is therefore mean convex,
and a closed surface with $H\le0$ would touch its smallest enclosing
sphere from inside, in contradiction with the maximum principle; in
particular there is no minimal surface and no apparent horizon.

\begin{figure}[t]
\centering\includegraphics[width=\linewidth]{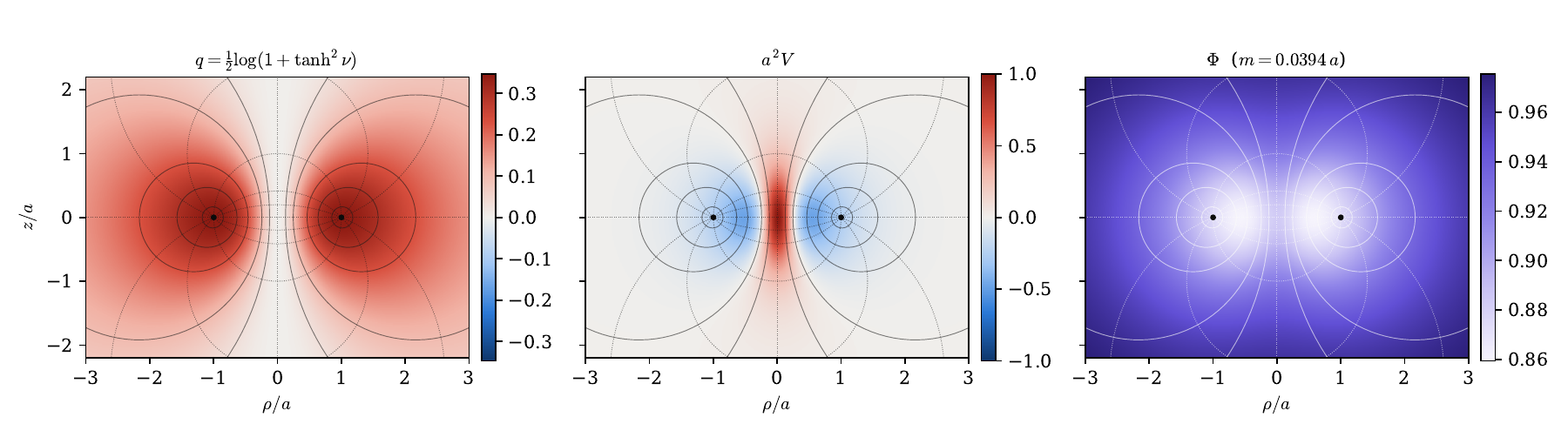}
\caption{The simplest explicit Brill wave \eqref{simplest} ($\ell=2$,
$k=0$) in a meridional plane: the seed function $q$, the potential
$a^2V$ (colour scale clipped at $\pm1$; $a^2V=0.92$ on the axis at
$z=0.2a$) and the conformal factor $\Phi$, with the mass from the
digamma series. Thin circles are tori $\nu=$ const, dotted circles
surfaces $u=$ const; the dots mark the circle $\rho=a$.}\label{figsimplest}
\end{figure}

\section{Curvature: the Kretschmann scalar}\label{sec_curvature}
The parameter $k$ and the values of $q$ are coordinate measures of the
strength of the wave. An invariant measure is provided by the curvature.
On a time-symmetric vacuum slice the magnetic part of the Weyl tensor
vanishes and its electric part is the Ricci tensor ${}^{3}\!R_{ij}$ of
$g$, so the four-dimensional Kretschmann scalar is
\begin{equation}\label{K8}
 \mathcal K\equiv R_{abcd}R^{abcd}=8\,{}^{3}\!R_{ij}\,{}^{3}\!R^{ij}
 =8\,(\lambda_1^2+\lambda_2^2+\lambda_3^2),
\end{equation}
where $\lambda_i$ are the eigenvalues of ${}^{3}\!R^i{}_j$, the tidal
field; $\lambda_1+\lambda_2+\lambda_3=0$ by the constraint. The slice is
the fixed-point set of the reflection $t\to-t$, so $\mathcal K$
restricted to it is an invariant of the space-time. For the
Schwarzschild slice, $\Phi=1+m/2r$ and $q=0$, Eq.~\eqref{K8} gives
$\mathcal K=48m^2/R_{\rm A}^6$ with $R_{\rm A}$ the areal radius, hence
$\mathcal Km^4=3/4$ on the horizon.

\subsection{The formula}
By \eqref{gtor} and \eqref{PhiF} the metric is
$g=a^2F^4[e^{2q}(\dd\nu^2+\dd u^2)+\sinh^2\nu\,\dd\varphi^2]$. Put
$\Lambda=\log F$. The conformal transformation of the Ricci tensor, with
\eqref{F_eq} used to eliminate $\Lambda_{\nu\nu}$, gives
\begin{equation}\label{Riccitor}
 {}^{3}\!R^\nu{}_\nu=\frac{\mathcal E_{\nu\nu}}{a^2F^4e^{2q}},\qquad
 {}^{3}\!R^u{}_u=\frac{\mathcal E_{uu}}{a^2F^4e^{2q}},\qquad
 {}^{3}\!R^\nu{}_u=\frac{\mathcal E_{\nu u}}{a^2F^4e^{2q}},\qquad
 {}^{3}\!R^\varphi{}_\varphi=\frac{\mathcal E_{\varphi\varphi}}{a^2F^4e^{2q}},
\end{equation}
with
\begin{equation}\label{Ecomp}
\begin{aligned}
 \mathcal E_{\nu\nu}&=q'(\coth\nu+2\Lambda_\nu)+2\Lambda_{uu}
     +4\Lambda_\nu^2+2\coth\nu\,\Lambda_\nu,\\
 \mathcal E_{\varphi\varphi}&=-\tfrac12+2w-2\coth\nu\,\Lambda_\nu
     -2\Lambda_\nu^2-2\Lambda_u^2,\\
 \mathcal E_{uu}&=-\mathcal E_{\nu\nu}-\mathcal E_{\varphi\varphi},\qquad
 \mathcal E_{\nu u}=-2\Lambda_{\nu u}+4\Lambda_\nu\Lambda_u+2q'\Lambda_u ,
\end{aligned}
\end{equation}
and therefore
\begin{equation}\label{Ktor}
 \boxed{\mathcal K=\frac{8\,\bigl(\mathcal E_{\nu\nu}^2+\mathcal E_{uu}^2
   +2\mathcal E_{\nu u}^2+\mathcal E_{\varphi\varphi}^2\bigr)}{a^4F^8e^{4q}}.}
\end{equation}
One eigenvalue is
$\lambda_\varphi=\mathcal E_{\varphi\varphi}/(a^2F^4e^{2q})$, with
eigenvector $\partial_\varphi$; the other two are
$[\half(\mathcal E_{\nu\nu}+\mathcal E_{uu})\pm
\{\tfrac14(\mathcal E_{\nu\nu}-\mathcal E_{uu})^2+\mathcal E_{\nu u}^2\}^{1/2}]
/(a^2F^4e^{2q})$. For flat space, $F=d^{-1/2}$ and $q=0$, all four
quantities vanish identically. Only $F$, $F_\nu$, $F_u$, $F_{uu}$ and
$F_{\nu u}$ enter; they are the Fourier sums of the modes
\eqref{explicit_modes} and of their first derivatives.

\subsection{Exact values}
\paragraph{Axis.}
On the axis $q=\half q_{\rho\rho}\rho^2+O(\rho^4)$ with
$q_{\rho\rho}=16\,w(0)\,a^2/(z^2+a^2)^2$ and $w(0)$ given by
\eqref{waxis}. With $\Phi_0(z)=\Phi(0,z)$, the eigenvalues are
$(\lambda_z,-\half\lambda_z,-\half\lambda_z)$ and
\begin{equation}\label{Kaxis}
 \lambda_z=-\frac1{\Phi_0^4}\Bigl[\frac{24\,w(0)\,a^2}{(z^2+a^2)^2}
   +\frac{2\Phi_0''}{\Phi_0}-\frac{4\Phi_0'^2}{\Phi_0^2}\Bigr],
 \qquad \mathcal K(0,z)=12\lambda_z^2 .
\end{equation}
The function $\Phi_0$ is supplied by Theorem~\ref{reconstruction}:
$\Phi_0=1+\sqrt{1-\cos u}\,\sum_nd_n(0)\cos(nu)$ with $z=a\cot(u/2)$ and
the coefficients \eqref{dnaxis}.

\paragraph{Circle.}
At the circle $f_n\simeq c_ne^{-(n+1/2)\nu}$ and
$q=q_\infty+\hat q\,e^{-2\nu}+O(e^{-4\nu})$, where, by
\eqref{explicit_modes}, \eqref{nusigma} and \eqref{wformula},
\begin{equation}\label{circleconst}
 c_n=\frac{\epsn\Gamma(\alpha_n)\Gamma(\beta_n)}{\pi\,n!}\,\kappa^{n/2},
 \qquad \kappa=\frac{4}{1-k}\,e^{-2\sqrt k\,\atanh\sqrt k},
 \qquad \hat q=\kappa\,(S+k)-1,
\end{equation}
and $q_\infty$ is \eqref{qinf}. Expanding \eqref{Ecomp} to second order
in $e^{-\nu}$ gives, with $b_n=c_n/c_0$,
\begin{equation}\label{Kcircle}
 \omega_0=b_1^2-1-\hat q,\qquad \omega_2=3b_1^2-4b_2,\qquad
 \boxed{\mathcal K_{\rm circle}
   =\frac{16\,(3\omega_0^2+\omega_2^2)}{a^4c_0^8\,e^{4q_\infty}},}
\end{equation}
and the eigenvalues
$(\lambda_\varphi,\lambda_\rho,\lambda_z)
=(-2\omega_0,\ \omega_0+\omega_2,\ \omega_0-\omega_2)/(a^2c_0^4e^{2q_\infty})$.
Apart from $q_\infty$ everything here is elementary:
\begin{equation}
 c_0=\frac1{\cos(\pi\sqrt S/2)},\qquad
 b_1=\frac{\sqrt\kappa\,s_1\cos(\pi\sqrt S/2)}{\sin(\pi s_1/2)},\qquad
 b_2=\frac{\kappa\,(1-s_2^2)\cos(\pi\sqrt S/2)}{4\cos(\pi s_2/2)},
\end{equation}
with $s_n=\sqrt{S+kn^2}$ as in \eqref{ab}. The terms that would be
singular at the circle cancel by the mode equation \eqref{modeeq} with
$n=0$.

\paragraph{Weak waves.}
By \eqref{qweak}--\eqref{phiweak},
$\Phi_0=1+k\,(\tfrac13-\tfrac15z^2/a^2)+O(kz^4,k^2)$, and
$w(0)=-k/2+O(k^2)$. Hence \eqref{Kaxis} gives at the centre
\begin{equation}\label{Kweak}
 \lambda_z(0,0)=\frac{64}{5}\,\frac k{a^2}+O(k^2),\qquad
 \mathcal K(0,0)=\frac{49152}{25}\,\frac{k^2}{a^4}+O(k^3).
\end{equation}
On the circle $\omega_0=-\tfrac32\pi k$ and
$\omega_2=(32-\tfrac{21}2\pi)k$ to first order, so that
\begin{equation}
 \mathcal K_{\rm circle}=(1024-672\pi+117\pi^2)\,\frac{k^2}{a^4}+O(k^3)
 \approx67.59\,\frac{k^2}{a^4},
\end{equation}
which is 29 times smaller than at the centre.

\paragraph{Critical limit.}
By Proposition~\ref{profile}, $A_k^8\,\mathcal K\to\mathcal K[g_\infty]$
locally uniformly as $k\uparrow k_c$, where $g_\infty$ is the limiting
metric \eqref{ginfty}, and $m/a=2\sqrt2A_k+O(1)$ by \eqref{masspole}.
Therefore, at a fixed point,
\begin{equation}\label{Klimit}
 \mathcal K\,a^4\Bigl(\frac ma\Bigr)^8\longrightarrow
 4096\,\hat{\mathcal K}_\infty,\qquad
 \hat{\mathcal K}_\infty=a^4\,\mathcal K[g_\infty].
\end{equation}
For $g_\infty$ the function $F=\sqrt{\sigma R/\sinh\nu}$ in \eqref{Ecomp}
does not depend on $u$. Thus $\partial_u$ is a Killing field of
$g_\infty$, $\mathcal E_{\nu u}=0$, and with
$\tau=(1-\sigma^2)(1-k+2k\sigma^2)/(\sigma R^3)$ and the function
$W_{k,S}$ of Section~\ref{sec_scaling},
\begin{equation}
 \mathcal E_{\nu\nu}=\tau\,(\tau+q'-\coth\nu),\qquad
 \mathcal E_{\varphi\varphi}=2W_{k,1}(\sigma)-\half\tau^2,\qquad
 \hat{\mathcal K}_\infty=\frac{8\,\bigl(\mathcal E_{\nu\nu}^2
   +\mathcal E_{uu}^2+\mathcal E_{\varphi\varphi}^2\bigr)\sinh^4\nu}
   {\sigma^4R^4e^{4q}}\bigg|_{k=k_c}.
\end{equation}
By \eqref{qprime} and \eqref{qclosed}, $q'-\coth\nu$ and
$e^{q}/\sinh\nu$ are functions of $\sigma$, so
$\hat{\mathcal K}_\infty$ is an explicit function of $\sigma$ alone. It
equals $12\,[\half-6w(0)]^2=4.9146\times10^4$ on the axis and
$3\kappa^2(1+k_c)^2e^{-4q_\infty}=4.7407\times10^4$ on the circle, with
the minimum $1.679\times10^4$ at $\sigma=0.591$.

\subsection{Results}
\begin{table}[tb]
\centering
\begin{tabular}{rrrrrr}
\toprule
$k$&$m/a$&$\mathcal K(0,0)\,a^4$&$\mathcal K_{\rm circle}\,a^4$
 &$\mathcal K(0,0)\,m^4$&$\mathcal K_{\rm circle}\,m^4$\\\midrule
$-3$&$0.23115$&$206.2$&$1127$&$0.5886$&$3.217$\\
$0.1$&$3.901\times10^{-3}$&$21.41$&$0.6305$&$4.960\times10^{-9}$&$1.461\times10^{-10}$\\
$0.3$&$0.060133$&$184.9$&$4.467$&$2.418\times10^{-3}$&$5.841\times10^{-5}$\\
$0.5$&$0.36598$&$236.2$&$6.579$&$4.237$&$0.1180$\\
$0.7$&$3.2016$&$6.324$&$0.6159$&$664.4$&$64.72$\\
$0.7262$&$4.9560$&$1.303$&$0.1855$&$786.4$&$111.9$\\
$0.7670$&$13.129$&$0.01176$&$3.746\times10^{-3}$&$349.4$&$111.3$\\
$0.79$&$38.650$&$1.252\times10^{-5}$&$7.405\times10^{-6}$&$27.93$&$16.52$\\
$0.80$&$127.01$&$2.030\times10^{-9}$&$1.656\times10^{-9}$&$0.5282$&$0.4309$\\
\bottomrule
\end{tabular}
\caption{Kretschmann scalar at the centre, from \eqref{Kaxis}, and on
the circle, from \eqref{Kcircle}, in units of $a$ and of the mass. The
row $k=0.7262$ is the maximum of $\mathcal K(0,0)\,m^4$; $k=0.7670$ is
$k_h$.}\label{tableK}
\end{table}
\begin{figure}[tb]
\centering\includegraphics[width=\linewidth]{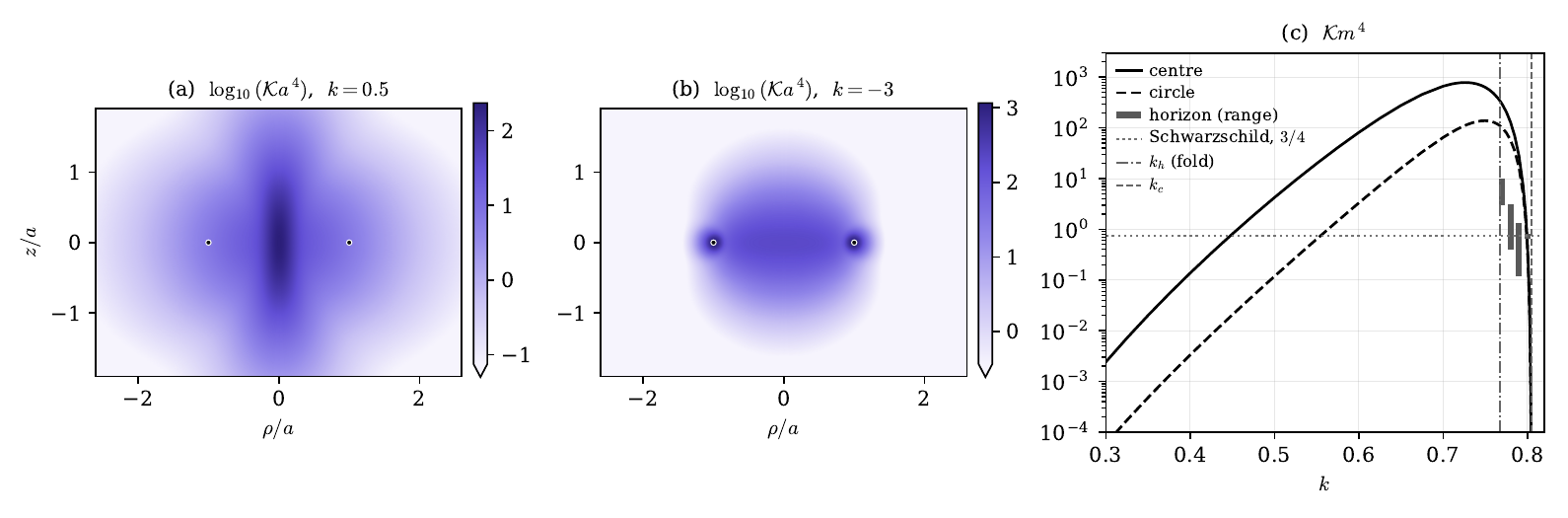}
\caption{The Kretschmann scalar. (a), (b) $\log_{10}(\mathcal Ka^4)$ in
a meridional plane, with $\rho$ continued to negative values across the
axis, for $k=0.5$, where the maximum is at the centre, and for $k=-3$,
where it lies next to the circle. Dots mark the circle; each panel has
its own scale, spanning 3.5 decades below its maximum. (c) The
dimensionless $\mathcal Km^4$ at the centre (solid) and on the circle
(dashed) as a function of $k$. The bars are its ranges over the outer
minimal spheres of Section~\ref{sec_trapped} at $k=0.77$, $0.78$, $0.79$
and $0.80$; the dotted line is the value $3/4$ on a Schwarzschild
horizon.}\label{figK}
\end{figure}
Table~\ref{tableK} lists the two exact values; in addition
\eqref{Ktor} was evaluated on a toroidal grid, with the modes summed to
$n=1200$. The following holds.

(i) For $0<k\lesssim0.799$ the maximum of $\mathcal K$ over the slice is
attained at the centre $\rho=z=0$, not on the circle
(Figure~\ref{figK}a). There
$\lambda_z>0$: the tidal field compresses along the axis and stretches in
the equatorial plane. $\mathcal K(0,0)\,a^4$ grows as $1966\,k^2$,
reaches $266.6$ at $k=0.432$ and then decreases, because
$\mathcal K\propto\Phi^{-8}$ and $\Phi(0,0)$ grows.

(ii) In units of the mass the central curvature keeps growing, up to
$\mathcal K(0,0)\,m^4=786$ at $k=0.726$, where $S=0.765$ and $m=4.96\,a$
(Figure~\ref{figK}c).
This is a thousand times the value $3/4$ on a Schwarzschild horizon, and
it is reached below the value $k_h$ at which the marginally trapped
spheres of Section~\ref{sec_trapped} first appear; at $k_h$ it is
still $349$.

(iii) Near criticality the wave zone is weakly curved in units of $m$:
by \eqref{Klimit}, $\mathcal Km^4\simeq4096\,\hat{\mathcal K}_\infty
[\pi(1-S)/8]^4\to0$ there, the wave region lying beyond a throat of
size $m\gg a$. The maximum of $\mathcal K$ is then on the apparent
horizon. On the outer minimal spheres of Section~\ref{sec_trapped},
$\mathcal Km^4$ ranges over $[2.99,9.91]$, $[0.41,3.16]$, $[0.12,1.35]$
and $[0.66,0.80]$ for $k=0.77$, $0.78$, $0.79$ and $0.80$,
approaching the Schwarzschild value $3/4$.

(iv) For $k<0$ the sign of the tidal field is reversed. For
$k\lesssim-1.8$ the maximum of $\mathcal K$ lies in the equatorial plane
at or just inside the circle (Figure~\ref{figK}b); for $k=-3$ it is
$1156/a^4$ at
$\rho=0.988\,a$, against $1127/a^4$ on the circle.

(v) The same formulas hold for the rescaled families of
Section~\ref{sec_scaling}, with the modes \eqref{modescaled}; in
\eqref{circleconst} then
$c_n=\epsn\Gamma(\tilde\alpha_n)\Gamma(\tilde\beta_n)\kappa^{n/2\ell}
/[\pi\sqrt\ell\,\Gamma(n/\ell+1)]$ and, for $\ell\ge2$, $\hat q=-1$. For
the simplest wave \eqref{simplest}, $w(0)=\tfrac14$ and
$\Phi(0,0)=0.88285$, and
\begin{equation}
 \mathcal K(0,0)=\frac{1289.86}{a^4},\qquad
 \mathcal K_{\rm circle}=\frac{3.6358}{a^4},
\end{equation}
the maximum being again at the centre; with $m=0.03939\,a$ this is
$\mathcal Km^4=3.1\times10^{-3}$, a weak wave by this measure as well.

Statements (i), (iii) and (iv) about the position of the maximum rest
on the numerical evaluation. The checks of \eqref{Ktor},
\eqref{Kaxis} and \eqref{Kcircle} are described in
Appendix~\ref{app_verification}.

\section{Discussion}
The construction provides complete analytic Brill data on $\R^3$ with
explicit hypergeometric Fourier modes and an explicit convergent mass
series. The free function is a closed-form combination of elementary functions
and a dilogarithm, with an elementary non-flat member at $k=-3$. In
summary, the solution consists of $q$ in \eqref{qclosed} and $\Phi$ in
\eqref{Phiseries} with \eqref{explicit_modes}, inserted into the metric
\eqref{metric}; its mass is \eqref{mass}. These
are exact finite-amplitude initial data. A ground-state identity determines the complete
physical parameter interval, and Fourier identification establishes the
global meaning of the series.

The distinct contribution is this solvable analytic realization and its
explicit normalizations, mass and critical parameter. The symmetry
reduction, general existence theory, the possibility of critical mass
blow-up and the formation of trapped surfaces near criticality
\cite{BeigOM} predate the present work. Explicit vacuum data with
non-trivial topology \cite{BrillLindquist,Misner} and analytic linearized
waves \cite{LinearBrill} have long been known. The present family differs in
being, to the author's knowledge, the first explicit real-analytic, fully
nonlinear vacuum data on $\R^3$ with a single asymptotically flat end.

Theorem~\ref{trapped} proves that these pure-radiation data contain
outer trapped surfaces and an apparent horizon for $k$ close to $k_c$,
with an explicit limiting profile \eqref{profilelimit}. The numerical
marginally trapped spheres indicate that trapped surfaces are present for
$k_h<k<k_c$, and they suggest a useful
benchmark for horizon finders, while the exact formulas provide benchmarks
for elliptic solvers of the kind used in the recent numerical
constructions \cite{BGW25,BGIMW26,AranhaOliveira,MastersonHobill}.
Section~\ref{sec_scaling} shows that the family is the $\ell=1$ member
of a sequence of solvable families $(\ell,k)$, which exhausts Natanzon's
hypergeometric class, and that the sequence contains an elementary
member, \eqref{simplest}; the critical amplitude varies along the
sequence. Remaining questions include an analytic lower bound on the
range of $k$ with trapped surfaces, the global geometry of trapped
surfaces, the branch of large negative $k$, solvable profiles outside the
hypergeometric class, and the evolution of these initial data. The coordinate radius $a$ is the radius
of the distinguished circle in the auxiliary Euclidean coordinates;
it is not a proper size of the wave packet.

\section*{Code availability}
The Python scripts described in Appendix~\ref{app_verification} are
provided as supplementary material; they are included as ancillary files
with the arXiv version of this article.

\section*{Use of AI tools}
Large language model assistants, including Claude (Anthropic), were used
in the exploration that led to the solvable family, in symbolic and
numerical checks, and in drafting and revising the text. The author has
reviewed the results and the text and takes full responsibility for the
content of this article.

\appendix
\section{Asymptotic control used in the reconstruction}\label{app_asym}
Here are details of the decay step in Theorem~\ref{existence}. For the
variational solution $\phi\in\D\subset L^6$, let $f=V(1+\phi)$.
The analytic potential is bounded and $V=O(r^{-4})$. In particular
$V\in L^1\cap L^{6/5}$ and $V\phi\in L^1$ by H\"older's inequality.
Local elliptic estimates first give smoothness. Rescaling annuli
$r\sim R$ to fixed size gives a uniformly elliptic equation with potential
$R^2V=O(R^{-2})$ and source $R^2V=O(R^{-2})$. The scaled $L^6$ norm of
$\phi$ is $R^{-1/2}\|\phi\|_{L^6(r\sim R)}$; interior estimates therefore
give $\phi=o(1)$ uniformly at infinity. Thus $f=O(r^{-4})$.

The Newton potential
\begin{equation}\label{Newton}
 \phi(x)=\frac1{4\pi}\int_{\R^3}\frac{f(y)}{|x-y|}\,\dd^3y
\end{equation}
represents the solution: its difference from $\phi$ is an entire harmonic
function in $\D$ and hence zero. Splitting this integral into
$|y|<r/2$, $|y-x|<r/2$, and the remaining region gives $\phi=O(r^{-1})$.
By the symmetries already following from uniqueness, $f(y)=f(-y)$.
Pairing $y$ and $-y$ cancels the dipole term of the Newton kernel.
On $|y|<r/2$, the paired remainder is bounded by
$C|y|^2/r^3$; its integral against $|f|$ is $O(r^{-2})$ because
$\int_{|y|<r/2}|y|^2|f|\,\dd^3y=O(r)$.
The exterior contribution is also $O(r^{-2})$. Therefore
\begin{equation}
 \phi=\frac{M}{r}+O(r^{-2}),\qquad
 M=\frac1{4\pi}\int f\,\dd^3y=\frac m2.
\end{equation}
Rescaled interior estimates give the first two derivative bounds in
\eqref{AF}. For the mass proof only this $O_2(r^{-2})$ estimate is needed.

The next term of the expansion, used in \eqref{cusp}, requires one more
step. By \eqref{V_asym} and the estimate just proved,
\begin{equation}\label{sourceexpansion}
 f=V\Phi=\frac{4a^2w(0)}{r^4}+\frac{2a^2m\,w(0)}{r^5}+O(r^{-6}).
\end{equation}
Let $\chi_1$ be a smooth radial cutoff which vanishes near the origin and
equals one for $r\ge a$, and write $f=f_{\rm rad}+\tilde f$, where
$f_{\rm rad}$ is $\chi_1$ times the two displayed terms. The Newton
potential of the radial function $f_{\rm rad}$ is radial, and since
$\Delta_3r^{-2}=2r^{-4}$ and $\Delta_3r^{-3}=6r^{-5}$ it equals
$c_1/r-2a^2w(0)/r^2-a^2m\,w(0)/(3r^3)$ for $r\ge a$. The remainder
$\tilde f$ is bounded, even under $y\mapsto-y$, and $O(r^{-6})$, so
$\int|y|^2|\tilde f|\,\dd^3y<\infty$. Repeating the multipole argument
above with one more order, its dipole term vanishes and its Newton
potential is $c_2/r+O(r^{-3})$. Hence, with $c_1+c_2=m/2$,
\begin{equation}\label{nextorder}
 \Phi=1+\frac m{2r}-\frac{2a^2w(0)}{r^2}+O(r^{-3}).
\end{equation}

\section{The energy boundary term in toroidal coordinates}
For clarity, a coordinate integration by parts must retain the circle
boundary term. If $\eta$ is axisymmetric and sufficiently regular,
with the decay of a function in the preceding arguments, and
$G=\eta/\sqrt d$, then
\begin{equation}\label{boundaryenergy}
 \begin{split}
 E_k[\eta]={}&a\int
 \left[G_\nu^2+G_u^2-\left(\frac14+w\right)G^2\right]
 \sinh\nu\,\dd\nu\dd u\dd\varphi\\
 &+2\pi^2a\,\eta(a,0)^2.
 \end{split}
\end{equation}
The final term comes from $\nu\to\infty$, which is a coordinate
boundary, not a physical boundary. For a general non-axisymmetric
function the bulk integrand additionally contains
$G_\varphi^2/\sinh^2\nu$, and the final term becomes
$\pi a\int_0^{2\pi}\eta(a\cos\varphi,a\sin\varphi,0)^2\dd\varphi$.
The Cartesian identity \eqref{groundstate} avoids this bookkeeping.

\section{Verification and reproducibility}\label{app_verification}
The accompanying Python files implement the analytic profile, the mode
functions, the mass, and the restricted shooting calculation; they
reproduce every number, table and figure of this article.
High-precision checks of \eqref{modeeq} and $q''/4=w$ were performed for
$n=0,1,3$, $\sigma=0.1,0.4,0.8$, and
$k=0.3,-1,0.78,-10$. With 45-decimal arithmetic, absolute residuals were
below $10^{-40}$. The closed form \eqref{qclosed} agrees with a direct
quadrature of \eqref{qprime} to 30 digits for $k=0.3,0.78,-1,-10$;
the elementary case \eqref{qminus3} and the circle value \eqref{qinf}
were checked in the same way.
The auxiliary-function identity \eqref{auxeq} was
checked at the same points, and \eqref{trialenergy} and
\eqref{boundaryenergy} by quadrature for $k=0.3,-1,0.78$. The convergence study of
Table~\ref{tableconvergence} and the sign of the mean curvature of the
displaced surfaces are produced by separate scripts. The constraint
itself was tested with the series \eqref{Phiseries}: the sum
$\Delta_3\Phi+V\Phi$, computed by finite differences in cylindrical
coordinates, is of order $10^{-11}$ at four points of the $(\rho,z)$
plane for $k=0.3$ and $k=-1$, while the two terms separately are of order
$10^{-2}$ to $1$. These checks supplement the derivations;
they do not replace the existence and reconstruction proofs.

The mass table uses \eqref{masstail} with $N=2000$. The expansion
\eqref{tail} is verified symbolically, and
\eqref{masspole} numerically: $m/a-8/[\pi(1-S)]$ tends to
$-6.96$ as $k\uparrow k_c$, with $S'(k_c)=4.047659$.
The first-order fields \eqref{phiweak} are checked symbolically against the
Poisson equation, and the coefficient in \eqref{weakmass} follows from
Beta integrals; independently, $[m(k)+m(-k)]/(2k^2)=0.3141603$ and
$0.31415928$ for $k=10^{-3}$ and $10^{-4}$, against
$\pi/10=0.31415927$. The ingredients of Section~\ref{sec_proof} are checked by a separate
script: the identity \eqref{zeromode} to 25 digits, positivity of the
modes and the comparison bound of Lemma~\ref{modebounds} on a sample of
$n\le40$, the two forms of \eqref{Tdef} and the identity
\eqref{spheresign} to 12 digits, and the values \eqref{sigmastar}. At
three points off the axis $\Phi_k-A_kh_k$ tends to finite limits, for
example $0.4661$ at $(\rho,z)=(2a,a)$, while $A_k$ grows from $23$ at
$k_c-k=10^{-2}$ to $2.2\times10^4$ at $k_c-k=10^{-5}$. Figure~\ref{figfields} is computed from the mode tables; at ten test
points for each of its two values of $k$, including points at distance
$10^{-3}a$ from the circle, the mapped values agree with the direct
summation of \eqref{Phiseries} to a relative accuracy of $10^{-7}$, and
they approach \eqref{Phicircle} at the circle. The other figures are generated from the analytic profiles,
the mass formula and the shooting calculation. No claim is made that a
finite set of numerical residuals certifies the entire continuum solution.

The statements of Section~\ref{sec_negative} are produced by
\texttt{negative\_branch.py}. The bracket in \eqref{sphereH} and the
torus criterion are evaluated on a grid with
$0.02\le\nu\le(\pi/2)\sqrt{|k|}+8$ and 61 values of $u$, for $k$ in steps
of $0.5$ down to $-8.5$ and for $k=-9,-10,-12,-15,-20,-25,-30$. The
fast evaluator is used for $\nu\le3$ and the modes
\eqref{explicit_modes} in multiple precision beyond; for $2\le\nu\le3$
the two agree to $2\times10^{-7}$ for $k\ge-12$ and to $4\times10^{-5}$
at $k=-30$. On a finer
grid near the circle the minimum of the bracket is $1.169$, $0.297$ and
$0.002$ for $k=-6,-8,-8.5$. The same script reproduces the outer
trapped spheres and tori of Section~\ref{sec_trapped} at $k=0.79$ and
$k=0.804$.

The rescaled families of Section~\ref{sec_scaling} are implemented in
\texttt{scaled\_family.py} (30-digit arithmetic) and in the fast
evaluator \texttt{brill\_l.py}. For $(\ell,k)=(2,0)$, $(2,\tfrac12)$,
$(3,-1)$ and $(3,0.3)$ the former checks $w_{\ell,k}=q_{\ell,k}''/4$ to
$10^{-30}$, the mode equation for the functions \eqref{modescaled} with
$n=0,1,2,3,7$ to $10^{-30}$, the axis normalization $\nu f_n'\to-\sqrt2\epsn/\pi$,
the identity \eqref{scalingid} to $10^{-12}$, and the constraint
$\Delta_3\Phi+V\Phi=0$ by finite differences to $10^{-10}$; the mass
from the digamma series agrees with the fast evaluator to $10^{-4}$
($0.039390$ against $0.039382$ for the member \eqref{simplest}), and
$\Phi$ from the two codes agrees to $10^{-9}$ at interior points. The
closed forms \eqref{simplest} and \eqref{qelem} for $\ell=3$ were
verified symbolically. Figure~\ref{figsimplest} is produced by
\texttt{make\_simplest\_maps.py}.

For Section~\ref{sec_curvature}, the Ricci tensor of \eqref{metric} was
derived by computer algebra in the toroidal and in the cylindrical form
(\texttt{derive\_formula.py}). The scalar curvature reproduces
\eqref{F_eq}, and for the Schwarzschild slice
$8\,{}^{3}\!R_{ij}\,{}^{3}\!R^{ij}=48m^2/R_{\rm A}^6$ exactly. Formula
\eqref{Ktor} agrees with the machine-derived tensor at 2000 random field
values to $7\times10^{-15}$, including the eigenvalues. Three
independent evaluations of $\mathcal K$ agree (\texttt{kretschmann.py}):
the mode sums in \eqref{Ktor}; the exact series in multiple precision,
differentiated numerically in $(\rho,z)$ and inserted in the
machine-derived cylindrical tensor, to $10^{-8}$--$10^{-9}$; and
fourth-order differences of the fast evaluator of $\Phi$, to
$10^{-6}$--$10^{-7}$. The axis formula \eqref{Kaxis} agrees with the
extrapolation of the mode-sum values to $\nu\to0$ to $10^{-7}$, and
\eqref{Kcircle} with the mode-sum values at $\nu=5$ to $7$ to
$10^{-5}$--$10^{-6}$, the size of the neglected $O(e^{-2\nu})$ terms.
The coefficient $64/5$ in \eqref{Kweak} is reproduced to six digits by
$[\lambda_z(k)-\lambda_z(-k)]/(2k)$ at $k=10^{-3}$. For \eqref{Klimit},
$\mathcal K_{\rm circle}\,a^4(m/a)^8/4096$ equals $3.558$, $4.607$,
$4.727$ and $4.739\times10^4$ at $1-S=10^{-2},\dots,10^{-5}$, and the
values on the apparent horizons obtained by finite differences agree
with spline interpolation of the mode-sum values to five digits.
Figure~\ref{figK} is produced by \texttt{make\_kretschmann\_figure.py}.

\end{document}